\documentclass[11pt]{article}

\usepackage[margin=1in]{geometry}
\usepackage[T1]{fontenc}
\usepackage{lmodern}
\usepackage{microtype}
\usepackage{amsmath,amssymb,amsthm}
\usepackage{graphicx}
\usepackage{booktabs}
\usepackage{natbib}
\usepackage{url}
\usepackage[hidelinks]{hyperref}
\usepackage{caption}
\usepackage{array}
\usepackage{placeins}
\usepackage{flafter}

\setcitestyle{semicolon}
\graphicspath{{figures/}}
\newcommand{\E}{\mathbb{E}}
\newcommand{\Prob}{\mathbb{P}}

\newcommand{\Ind}{\mathbf{1}}

\newtheorem{theorem}{Theorem}
\newtheorem{lemma}{Lemma}
\newtheorem{proposition}{Proposition}
\newtheorem{assumption}{Assumption}
\theoremstyle{definition}

\hypersetup{
  pdftitle={Bayesian Robustness Values for Modern Causal Panel Estimators via Riesz Representations},
  pdfauthor={Makoto Nakakita and Takahiro Hoshino},
  pdfsubject={Sensitivity analysis for causal panel estimators}
}

\begin{document}

\title{Bayesian Robustness Values for Modern Causal Panel Estimators via Riesz Representations}
\author{Makoto Nakakita\\RIKEN\\\texttt{makoto.nakakita@riken.jp}\and Takahiro Hoshino\\Keio University and RIKEN}
\date{}
\maketitle
\vspace{-1.0em}

\begin{abstract}
We develop a sensitivity-analysis workflow for causal panel estimators, covering synthetic difference-in-differences, matrix completion, fixed-effect imputation, and group-time average treatment effects. The workflow combines Riesz-representation omitted-variable-bias bounds with partial-$R^2$ robustness values and separates two reporting routes. Route A gives a direct sensitivity profile for additive or projected confounding summarized by outcome-side and Riesz-side partial $R^2$ values. Route B treats observed-covariate benchmarks as auxiliary data only when benchmark-count, alpha-side alignment, model-check, dependence, and dominance diagnostics are credible; otherwise its main role is demotion. We derive estimator-specific Riesz diagnostics and clarify which are fixed-weight, target-level, or first-stage-conditional rather than full derivatives of regularized training maps. Monte Carlo stress tests distinguish calibrated benchmark settings from dominance failure, coarse alpha-side benchmarks, benchmark dependence, noisy covariates, and concentrated SDID weights. In the California tobacco-control panel, the SDID estimate is $-15.60$ packs per capita; corrected finite-donor placebo inference gives standard error 9.49 and add-one $p=0.051$. A refit-weight finite-difference audit changes the Route A nullification robustness value from 0.054 to 0.045, leaving the low-single-digit conclusion unchanged. A county-level minimum-wage application applies the same profile to a multi-cohort staggered panel.
\end{abstract}

\noindent{\it Keywords:} omitted variable bias; synthetic difference-in-differences; matrix completion; comparative case study; partial identification.

\newpage

\section{Introduction}
\label{sec:intro}

Causal panel estimators have become central tools in applied causal inference. Synthetic difference-in-differences (SDID) of \citet{ArkhangelskyEtAl2021} bridges synthetic control and difference-in-differences, delivering double-robustness against violations of either parallel trends or perfect synthetic-control match. Matrix completion (MC) of \citet{AtheyEtAl2021} imputes counterfactual outcomes via nuclear-norm regularization, exploiting low-rank structure in the panel. The group-time average treatment effect framework of \citet{CallawaySantAnna2021} and the imputation-based approach of \citet{BorusyakJaravelSpiess2024} provide flexible identification under staggered treatment adoption. These estimators are often motivated by latent-factor and interactive-fixed-effect failures of simple two-way fixed effects \citep{Bai2009,Xu2017}. This paper therefore sits at the intersection of modern DiD heterogeneity work \citep{DeChaisemartinDHaultfoeuille2020,GoodmanBacon2021,RothSantAnnaBilinskiPoe2023}, synthetic-control extensions \citep{DoudchenkoImbens2016,BenMichaelFellerRothstein2021}, and sensitivity analysis for observational causal inference \citep{RosenbaumRubin1983,ImbensWooldridge2009,Manski1990,VanderWeeleDing2017,FermanPinto2021}.

Despite these advances, sensitivity analysis for modern panel estimators remains uneven. SDID and matrix completion still lack an omitted-variable-bias-bound-based robustness-value analysis, while recent work on staggered $ATT(g,t)$ is frequentist rather than posterior predictive. Consider a concrete applied scenario. A researcher estimates the effect of California's tobacco-control program on cigarette consumption using SDID and reports $\widehat\tau\approx -15.6$ packs per capita, a large negative point estimate. A natural robustness question is whether this conclusion survives omitted state-level health-policy initiatives that correlate with both treatment timing and outcome trajectory. The researcher has no estimator-native answer. \citet{CinelliHazlett2020}'s OLS robustness-value framework does not extend directly to SDID; \citet{BachKlaassenKueckMattesSpindler2025}'s frequentist double-machine-learning extension covers $ATT(g,t)$ but not SDID; and \citet{LiuYamamoto2025}'s Bayesian framework is parameterized in latent-confounder space rather than in omitted-variable-bound and partial-$R^2$ robustness-value units.

This paper develops a sensitivity workflow for this setting. We combine three ingredients. First, the Riesz-representation omitted-variable-bias framework of \citet{ChernozhukovEtAl2026} provides a unified machinery for bounding omitted-variable bias of general causal estimands. Its use of fitted representers is connected to double/debiased machine learning and automatic Riesz-representer estimation \citep{ChernozhukovEtAl2018,ChernozhukovNeweySingh2022}. Second, the partial-$R^2$ robustness-value framework of \citet{CinelliHazlett2020,CinelliHazlett2025} provides a scale-free and interpretable benchmarking device. This scale connects to selection-on-observables calibration \citep{Imbens2003,Oster2019,Frank2000}, weighted-estimator sensitivity \citep{WainsteinHazlett2025}, relative-correlation and partial-identification approaches \citep{Krauth2016,MastenPoirier2018,MastenPoirierZhang2024}, endogenous-control sensitivity \citep{DiegertMastenPoirier2022}, and local-misspecification sensitivity \citep{BonhommeWeidner2022}. Third, Bayesian inference on the partial-$R^2$ sensitivity parameters provides a natural language for uncertainty quantification when those parameters are not themselves data-identifiable.

The third ingredient requires care. Priors placed directly on the sensitivity parameters are useful for transparent robustness profiling, which we call Route A, but they do not use observed covariate benchmarks as data. Our auxiliary benchmark formulation, Route B, treats observed-covariate partial-$R^2$ pairs as draws from a benchmark population and updates the resulting omitted-variable-bias bound only when diagnostic checks support that modeling step. When the benchmark set is too small, too coarse, poorly aligned with the estimator's Riesz representer, or unable to dominate the hidden-confounder strength distribution, Route B is demoted to an exploratory stress test and Route A remains the primary analysis.

Our contribution differs from the closest concurrent literatures in four ways. First, relative to cross-fitted DML work on $ATT(g,t)$, we add closed-form fixed-weight SDID and target-level MC Riesz diagnostics on the partial-$R^2$ omitted-variable-bias scale. The $ATT(g,t)$ representer is adapted from \citet{BachKlaassenKueckMattesSpindler2025}; the SDID, MC, and BJS diagnostics and the Route A/B reporting layer are the paper's main methodological additions. Second, relative to weighted partial-$R^2$ OVB work, our fixed-weight SDID diagnostic is the panel-contrast analogue of a weighted sensitivity calculation when weights are treated as fixed, but the fitted-weight feedback is not hidden: the refit audit reports both direction and decision-scale magnitude. Third, relative to latent-confounder Bayesian panel sensitivity, the route diagnostic can demote an observed-benchmark update when the benchmark population is too small, too coarse, alpha-side degenerate, or dependent in a way that leaves little effective information. Fourth, relative to parallel-trends-prior Bayesian DiD approaches such as \citet{HanMitraHettingerOganisian2025}, the output is a decision-scale translation of frequentist breakdown thresholds into Route A prior probabilities and, only when diagnostics pass, Route B posterior-predictive probabilities.

In applications, the minimal reporting set is the fitted effect, sampling uncertainty, Riesz scale, Route A nullification and significance robustness values, benchmark count, alpha-side support, auxiliary-model check, dominance rationale, and route status.

The paper makes three contributions. First, it gives a Route A/Route B sensitivity workflow: Route A reports direct robustness-value profiles and prior-probability translations, while Route B is a diagnostic updating layer that uses observed-covariate benchmarks only after benchmark-count, alpha-side, dominance, dependence, and model-check diagnostics support calibration. Second, it derives estimator-specific Riesz diagnostics for SDID, matrix completion, fixed-effect imputation, and staggered $ATT(g,t)$, always stating whether the diagnostic is conditional on fitted weights, a target-level functional, or a first-stage fill-in operator. Third, it combines contraction and coverage theory, simulation stress tests, and two empirical studies: a single-treated-unit tobacco-control panel and a county-level staggered-adoption panel. The empirical studies show how the route diagnostics prevent observed-covariate benchmarks from being interpreted as calibrated posterior evidence when the benchmark population is not aligned with the estimator-specific Riesz representer.

Sections \ref{sec:framework}--\ref{sec:bayes} set up the omitted-variable-bias bound and the Bayesian routes. Section \ref{sec:riesz} gives the panel-estimator Riesz diagnostics. Section \ref{sec:asymptotic} gives theory and operating characteristics. Sections \ref{sec:application}--\ref{sec:staggered-app} give the two empirical studies. Section \ref{sec:discussion} gives reporting recommendations and limitations.

\section{Omitted-Variable-Bias Bounds via Riesz Representation}
\label{sec:framework}

\subsection{Setup and Riesz-Representation Bound}
\label{sec:framework:bound}

Let $W=(Y,D,X,U)$ denote the data, where $Y$ is the outcome, $D$ the treatment indicator, $X$ a vector of observed pretreatment covariates, and $U$ an unobserved confounder. The causal estimand $\theta_{\mathrm{long}}$ is identifiable under the long regression specification
\begin{equation}
\theta_{\mathrm{long}}=\E[m(W;g_{\mathrm{long}})],\qquad
 g_{\mathrm{long}}(D,X,U)=\E[Y\mid D,X,U],
\label{eq:long}
\end{equation}
where $m$ is a linear functional of the conditional expectation $g_{\mathrm{long}}$. The short specification accessible to the researcher replaces $g_{\mathrm{long}}$ by
\begin{equation}
 g_{\mathrm{short}}(D,X)=\E[Y\mid D,X].
\label{eq:short}
\end{equation}
The omitted-variable-biased estimand $\theta_{\mathrm{short}}$ generally differs from $\theta_{\mathrm{long}}$ by an amount that depends on the strength of $U$'s correlation with $Y$ and with the Riesz representer of the functional $m$ at $g_{\mathrm{short}}$.

The fundamental bound is
\begin{equation}
\left|\theta_{\mathrm{short}}-\theta_{\mathrm{long}}\right|^2
\leq
\E\!\left[\alpha^2(W)\sigma^2_{Y\mid D,X}(W)\right]
\frac{R^2_{Y\sim U\mid D,X}\,R^2_{\alpha\sim U\mid X}}{1-R^2_{\alpha\sim U\mid X}},
\label{eq:ovb}
\end{equation}
where $\alpha(W)$ is the Riesz representer of $m$ at $g_{\mathrm{short}}$ under the empirical-distribution inner product and $\sigma^2_{Y\mid D,X}$ is the residual variance of $Y$ given $(D,X)$. We write $R^2_{yu}$ and $R^2_{\alpha u}$ for the two partial $R^2$ terms and abbreviate
\[
M=\left\{\E[\alpha^2(W)\sigma^2_{Y\mid D,X}(W)]\right\}^{1/2}.
\]
The bound factorizes the source of omitted-variable bias into an intrinsic scale $M$ depending on the estimator and the data, and a sensitivity term governed by the two partial $R^2$ values of the confounder.

\subsection{Partial-\texorpdfstring{$R^2$}{R2} Interpretation and Robustness Values}
\label{sec:rv-scope}
\label{sec:rv}

The bound targets omitted components that can be summarized by their partial association with the outcome residual and with the Riesz representer. This class includes scalar or vector covariates $U$ after residualization, and it also includes projected latent components such as an interactive-factor term $U_{it}=\lambda_i'f_t$ when the researcher is willing to summarize the remaining unbalanced factor component by $(R^2_{yu},R^2_{\alpha u})$. Arbitrary unbalanced interactive fixed effects can change the counterfactual surface in ways that require a design-level factor argument or an explicit projection into this partial-$R^2$ confounding class.

This distinction matters for interpretation. A low nullification RV in the California tobacco-control study means that a small equal-strength partial-$R^2$ confounder, measured in the outcome-residual and fitted-Riesz directions, can move the fitted contrast. Its scope is the additive or projected omitted-component class summarized by the two partial-$R^2$ coordinates. Latent-factor imbalance, SUTVA violations, spillovers, and misspecified treatment timing require a design-level argument or a projection of the violation into that class. The point of the workflow is therefore to state the target confounding class, report the Riesz scale, and then use route diagnostics to avoid turning observed-covariate benchmarks into automatic posterior validation.

We work with the Riesz-scaled effect size
\begin{equation}
K=\frac{|\theta_{\mathrm{short}}|}{M},\qquad
M=\left\{\E[\alpha^2(W)\sigma^2_{Y\mid D,X}(W)]\right\}^{1/2}.
\label{eq:K}
\end{equation}
For inference, let $c_t$ denote the critical value of the estimator's sampling distribution, typically $z_{1-\alpha/2}=1.96$. The corresponding critical value on the $K$ scale is
\begin{equation}
c=c_t\,(SE/M).
\label{eq:c}
\end{equation}
For OLS, $SE/M=1/\sqrt{df}$, recovering the familiar Cinelli-Hazlett scaling. For SDID and MC, $M$ is computed from the Riesz representers of Section \ref{sec:riesz}, so the formula does not require residual degrees of freedom for a regularized estimator.

The robustness value $RV$ is the smallest equal-strength partial-$R^2$ value $r=R^2_{yu}=R^2_{\alpha u}\in[0,1]$ at which the OVB-adjusted effect size hits the threshold $q c$, where $q=1$ reaches the boundary of conventional significance and $q=0$ nullifies the point estimate. Under equal strength, the squared scaled bound is $r^2/(1-r)$. Solving $r^2/(1-r)=\widetilde K^2$ with $\widetilde K=\max(K-qc,0)$ gives
\begin{equation}
RV(K,c;q)=\frac{1}{2}\left\{\sqrt{\widetilde K^4+4\widetilde K^2}-\widetilde K^2\right\},
\qquad \widetilde K=\max(K-qc,0).
\label{eq:rv}
\end{equation}
The interpretation is invariant across estimators: an unobserved confounder must explain at least an $RV$ fraction of residual variance in both $Y$ and the Riesz representer $\alpha$ to alter the conclusion at the chosen benchmark.

\medskip
\noindent\emph{Connection to partial identification.} For any fixed equal-strength confounding scale $r$, the OVB bound defines a partial-identification set around the short-regression estimand, up to the sampling component used for inference. The nullification RV is the smallest $r$ at which zero enters this set; the significance RV is the corresponding threshold for a sampling-adjusted set. Route A therefore does not identify the long-regression effect. It places a transparent probability model over the radius of a partial-identification set, while Route B updates that radius only when the observed benchmark population is credible.

\section{Bayesian Reporting: Routes A and B}
\label{sec:bayes}

We develop two routes for placing a Bayesian framework on top of the Riesz-representation OVB machinery. Route A places direct priors on $(R^2_{yu},R^2_{\alpha u})$ without using observed-covariate information. Route B places a model on the population distribution of covariate strengths and updates via observed-covariate benchmarks. Route A is the default reporting baseline because it is always defined; Route B is reported as calibrated only after benchmark diagnostics support the auxiliary modeling step.

\noindent\emph{Route A direct prior profile.} Route A places Beta priors
\begin{equation}
R^2_{yu}\sim \mathrm{Beta}(a_{yu},b_{yu}),\qquad
R^2_{\alpha u}\sim \mathrm{Beta}(a_{\alpha u},b_{\alpha u}).
\label{eq:beta-priors}
\end{equation}
The prior parameters are set by the user. Observed covariates can still help choose a prior scale, for example through the heuristic that an unobserved confounder should be no stronger than the strongest observed covariate, but Route A should be interpreted as a prior-predictive sensitivity profile rather than a data-identified posterior analysis.

\medskip
\noindent\emph{Observed-covariate calibration and Route B.}\label{sec:bayes:routeB} For each observed covariate $X_j$, compute a bivariate partial-$R^2$ benchmark
\[
 m_{\mathrm{obs}}^j=(R^2_{Y,X_j\mid rest},R^2_{\alpha,X_j\mid rest})
\]
using the most granular covariate representation available. We use one pessimism map throughout. For $r\in(0,1)$ and $\kappa\geq1$, define the odds-scale multiplier
\begin{equation}
g_\kappa(r)=\operatorname{logit}^{-1}\{\operatorname{logit}(r)+\log\kappa\}
=\frac{\kappa r}{1-r+\kappa r}.
\label{eq:odds-link}
\end{equation}
A Route A calibration sets the prior means to
\begin{equation}
\E(R^2_{yu})=g_\kappa\!\left(\max_j R^2_{Y,X_j}\right),\qquad
\E(R^2_{\alpha u})=g_\kappa\!\left(\max_j R^2_{\alpha,X_j}\right),
\label{eq:kappa-cal}
\end{equation}
and fixes the Beta concentration after this mean calibration. We use concentration 20 as the reference value and report concentration sensitivity.

Measurement error has two distinct effects. Attenuated benchmark strengths can make a conditional Route B update look too robust, whereas the route gate works in the opposite direction by demoting weak observed alpha-side support. The full-pipeline experiment regenerates noisy covariates, residualizes them, and recomputes the partial-$R^2$ pairs. When reliability falls from 1.0 to 0.5, the median observed-to-latent alpha benchmark ratio falls to 0.812, the alpha-support gate rate falls from 0.963 to 0.819, and conditional coverage among promoted replications falls from 0.983 to 0.945. These results quantify both attenuation and the protective, but incomplete, response of the gate.

Route B turns the benchmarking heuristic into an auxiliary population model. Let $T_0(\theta)$ denote a mean, tail-quantile, or upper-support functional of the benchmark population. The same odds multiplier is applied coordinatewise:
\begin{equation}
T(\theta;\kappa)=\min\{g_\kappa(T_0(\theta)),1-\eta_{\mathrm{cap}}\},
\label{eq:link}
\end{equation}
where the minimum is coordinatewise and $\eta_{\mathrm{cap}}>0$ keeps the OVB denominator away from zero. Independent Beta marginals provide the baseline working family, with logit-normal and Gaussian-copula alternatives used as sensitivity checks.

Observed benchmark pairs share outcomes, residualization steps, and often correlated covariates, so a literal product likelihood can overstate information. We therefore treat the auxiliary update as a generalized or power posterior in the sense of generalized Bayes and robust/coarsened Bayes \citep{BissiriHolmesWalker2016,GrunwaldVanOmmen2017,MillerDunson2019}. The tempered working likelihood is
\begin{equation}
L_c(\theta;m^1_{\mathrm{obs}},\ldots,m^p_{\mathrm{obs}})
=\left\{\prod_{j=1}^p f(m^j_{\mathrm{obs}};\theta)\right\}^{\omega_p},
\qquad \omega_p=\frac{p_{\mathrm{eff}}}{p}.
\label{eq:auxlik}
\end{equation}
Here $p_{\mathrm{eff}}\leq p$ is an effective benchmark count, so the composite score has information of order $p_{\mathrm{eff}}$ rather than $p$. In an equicorrelation benchmark-score model with common correlation $\rho_b$, $p_{\mathrm{eff}}=p/\{1+(p-1)\rho_b\}$ follows from the variance of the average score. In non-equicorrelated designs we use the same mean-score formula, $p_{\mathrm{eff}}=p^2/({\bf 1}'\widehat R{\bf 1})$, where $\widehat R$ is the empirical benchmark-score correlation matrix after residualization; a spectral participation-ratio diagnostic, $(\operatorname{tr}\widehat R)^2/\operatorname{tr}(\widehat R^2)$, is reported as a secondary dimension-count check. These adjustments are modeling choices, not identification assumptions. If dependence leaves $p_{\mathrm{eff}}$ bounded, Route B does not contract as more redundant benchmarks are added. Updating with $L_c$ gives $\pi(\theta\mid m_{\mathrm{obs}})$, and the posterior of $(R^2_{yu},R^2_{\alpha u})$ is the pushforward through $T(\cdot;\kappa)$.

\medskip
\noindent\emph{Posterior summaries and route choice.} For each posterior draw $(R^{2,(s)}_{yu},R^{2,(s)}_{\alpha u})$, compute the OVB bound
\begin{equation}
B^{(s)}=M\left\{\frac{R^{2,(s)}_{yu} R^{2,(s)}_{\alpha u}}{1-R^{2,(s)}_{\alpha u}}\right\}^{1/2},
\label{eq:post-bound}
\end{equation}
and the worst-direction adjusted estimate
\begin{equation}
\widehat\theta^{(s)}_{adj}=\mathrm{sign}(\widehat\theta)\max\{|\widehat\theta|-B^{(s)},0\}.
\label{eq:adj-est}
\end{equation}
The posterior distribution of the robustness value is obtained by applying \eqref{eq:rv} to each draw of the adjusted Riesz-scaled effect size.

Route choice is a reporting rule, not an estimator-selection rule. Route A should always be reported. Route B is promoted from exploratory to calibrated only when the effective number of benchmark covariates is large enough, alpha-side benchmarks are nondegenerate and computed at the same observational level as the Riesz representer, posterior-predictive checks do not reject the auxiliary benchmark model, and the predictive distribution is plausibly pessimistic enough to dominate the hidden-confounder strength. In the simulations and applications, rejection of the marginal benchmark-model check at the 10\% level is used as the default M1 warning threshold; such a warning demotes Route B unless a design-specific justification is stated. A larger $\kappa$ is a sensitivity display rather than a calibration repair when the benchmark model is not supported.

\medskip
\noindent\emph{Bayesian interpretation.} Route A probabilities are prior probabilities over sensitivity parameters, not posterior probabilities about an identified causal effect. Route B probabilities are posterior-predictive probabilities under an auxiliary benchmark-population model. They are reported as calibrated only when benchmark count, alpha-side alignment, model-check diagnostics, and predictive-dominance plausibility all pass. If any diagnostic fails, Route B remains an exploratory stress profile and Route A is primary.

\section{Riesz Diagnostics for Causal Panel Estimators}
\label{sec:riesz}

Before giving formulas, we fix the scope of each diagnostic. Each diagnostic is conditional on a first-stage object. The SDID diagnostic is the fitted-contrast Riesz vector evaluated at the estimated unit and time weights. The matrix-completion diagnostic is the treated-cell target vector for a fixed treated-cell index set. The fixed-effect-imputation diagnostic is the imputation-operator Riesz vector conditional on the untreated-cell two-way projection. The group-time diagnostic is the fitted score contrast given the nuisance scores and comparison sets. These objects put several fitted causal-panel contrasts on a common omitted-variable-bias scale while keeping the first-stage convention visible.

\subsection{Estimator-Specific Representers}

\noindent\emph{Fixed-weight SDID Riesz diagnostic.} The SDID estimator \citep{ArkhangelskyEtAl2021} computes an ATT using unit weights $\widehat\omega_i$ and time weights $\widehat\lambda_t$. Let $g(i)\in\{C,T\}$ denote treatment-group membership and $p(t)\in\{pre,post\}$ denote period status. Conditional on the estimated weights being held fixed, the fitted SDID contrast is linear in the outcome surface.

\begin{proposition}[Fixed-weight SDID contrast representer]
Conditional on the estimated unit and time weights being held fixed, the Riesz representer for the fitted SDID contrast under the empirical-distribution inner product is
\begin{equation}
\alpha_{\mathrm{SDID}}(i,t)=NT\,\widehat\omega_i\widehat\lambda_t\frac{\sigma_{g(i),p(t)}}{W_{g(i),p(t)}},
\label{eq:sdid-rep}
\end{equation}
where $\sigma_{g,p}=+1$ for $(T,post)$ and $(C,pre)$ and $-1$ for $(C,post)$ and $(T,pre)$, and $W_{g,p}=(\sum_{i\in g}\widehat\omega_i)(\sum_{t\in p}\widehat\lambda_t)$.
\end{proposition}

The proof is immediate: apply the weighted double-difference to a generic outcome surface and match coefficients under the empirical cell inner product. The $NT$ prefactor cancels the empirical-mass normalization. Synthetic panels and the California tobacco-control panel verify the identity to machine precision.

If the weights were externally fixed, this diagnostic would be the panel-contrast analogue of a weighted partial-$R^2$ omitted-variable-bias calculation such as \citet{WainsteinHazlett2025}. The difference is that SDID weights are estimated from the same outcome panel. We therefore interpret (\ref{eq:sdid-rep}) as a fitted-contrast diagnostic and explicitly report a refit finite-difference audit; the audit is the device that quantifies the first-stage weight feedback omitted by the fixed-weight representation. A fully orthogonalized SDID sensitivity diagnostic is conceptually possible, in the spirit of orthogonal-score DiD and DML constructions \citep{ChernozhukovEtAl2018,SantAnnaZhao2020,BachKlaassenKueckMattesSpindler2025}. We do not use it as the baseline because the two-sided simplex constraints in SDID weights can generate nonstandard corner-solution asymptotics; the fixed-weight diagnostic plus refit audit is therefore a deliberate choice for interpretability and transparent measurement of weight feedback.

\noindent\emph{Matrix-completion target-level diagnostic.} The MC estimator imputes the counterfactual outcome surface for treated cells via nuclear-norm regularization. The diagnostic isolates the treated-cell counterfactual-mean target, with the nuclear-norm training map treated as the first-stage procedure that defines the imputed surface.

\begin{proposition}[Matrix-completion target-level representer]
For the treated-cell index set $\mathcal M$, the Riesz representer for the counterfactual-mean target underlying the MC ATT is
\begin{equation}
\alpha_{\mathrm{MC}}(i,t)=\frac{NT}{|\mathcal M|}\Ind\{(i,t)\in\mathcal M\}.
\label{eq:mc-rep}
\end{equation}
\end{proposition}

This representer is constant on the treated post-treatment cells and zero elsewhere. It is intentionally target-level: it bounds omitted-variable bias in the treated-cell counterfactual target conditional on the imputed surface. This target-level quantity is an interpretable scale for comparing MC with concentrated SDID weights, while a fully differentiated nuclear-norm training-map diagnostic is a distinct estimator-level object.

\noindent\emph{Group-time and fixed-effect-imputation representers.} For staggered designs, we record the $ATT(g,t)$ representer for the doubly robust score of \citet{SantAnnaZhao2020}, following the Riesz representation in \citet{BachKlaassenKueckMattesSpindler2025}.

\begin{proposition}[Group-time $ATT(g,t)$ fitted-score representer]
With comparison group $C$ and cross-fitted propensity $\widehat\pi(X_i)$,
\begin{equation}
\begin{aligned}
\alpha_{g,t}(W)=NT\Bigg\{&\frac{\Ind[i\in g]}{\Prob(i\in g)}-
\frac{\widehat\pi(X_i)\Ind[i\in C]}{(1-\widehat\pi(X_i))\Prob(i\in C)}\Bigg\}\\
&\times\{\Ind[s=t]-\Ind[s=g-1]\}.
\end{aligned}
\label{eq:att-rep}
\end{equation}
As with fixed-weight SDID, this is interpreted conditionally on the fitted first-stage functions.
\end{proposition}

The fixed-effect imputation estimator of \citet{BorusyakJaravelSpiess2024} estimates unit and period effects on untreated cells and imputes treated-cell counterfactuals from the fitted two-way fixed-effect structure. Let $\mathcal U$ and $\mathcal W$ denote untreated and treated cells, $N_1=|\mathcal W|$, and $N_{\mathrm{cell}}=NT$. Let
\[
\Pi=X_{\mathcal W}(X_{\mathcal U}'X_{\mathcal U})^{-}X_{\mathcal U}',
\qquad \widehat Y_{\mathcal W}(0)=\Pi Y_{\mathcal U}.
\]
The BJS ATT is $\widehat\tau_{\mathrm{BJS}}=N_1^{-1}{\bf 1}_{\mathcal W}'(Y_{\mathcal W}-\Pi Y_{\mathcal U})$.

\begin{proposition}[Fixed-effect imputation representer]
Conditional on the untreated-cell imputation operator $\Pi$, the empirical-distribution Riesz representer for $\widehat\tau_{\mathrm{BJS}}$ is
\begin{equation}
\alpha^{\mathrm{BJS}}_{it}=N_{\mathrm{cell}}/N_1\quad\text{for }(i,t)\in\mathcal W,
\qquad
\alpha^{\mathrm{BJS}}_{js}=-\frac{N_{\mathrm{cell}}}{N_1}(\Pi'{\bf 1}_{\mathcal W})_{js}
\quad\text{for }(j,s)\in\mathcal U.
\label{eq:bjs-rep}
\end{equation}
Its Riesz second moment is
\begin{equation}
\E[(\alpha^{\mathrm{BJS}})^2]=N_{\mathrm{cell}}\left\{\frac{1}{N_1}+\|N_1^{-1}\Pi'{\bf 1}_{\mathcal W}\|^2\right\}.
\label{eq:bjs-second}
\end{equation}
The representer annihilates additive unit and period shifts, so the diagnostic responds to deviations from the two-way fixed-effect imputation structure rather than to the fixed effects themselves.
\end{proposition}

\subsection{Comparison of Riesz Scales}

Table \ref{tab:riesz-real} compares the Riesz second moments of MC, BJS, and fixed-weight SDID on synthetic and California tobacco-control panels. The MC second moment equals $NT/|\mathcal M|$ and therefore depends only on panel dimensions and treated target cells. The BJS and SDID moments also depend on the untreated-cell imputation operator and fitted SDID weights. The table reports $\E[\alpha^2]$ on the empirical-distribution scale; the Route A scale $M$ used in robustness values is $\{\E[\alpha^2\sigma^2_{Y\mid D,X}]\}^{1/2}$ and therefore also incorporates the residual-variance factor.

\begin{center}
\begin{minipage}{0.98\textwidth}
\centering
\captionof{table}{Riesz-representer second moments and numerical verification.}
\label{tab:riesz-real}
\begin{tabular}{lcc}
\toprule
Quantity & Synthetic check & California panel \\
\midrule
\multicolumn{3}{l}{\emph{Second moments on the empirical-distribution scale}}\\
$\E[\alpha^2_{\mathrm{MC}}]$ & 100.75 & 100.75 \\
$\E[\alpha^2_{\mathrm{BJS}}]$ & 16.00 & 168.71 \\
$\E[\alpha^2_{\mathrm{SDID}}]$ & 281.45 & 567.83 \\
SDID/MC ratio & $2.79\times$ & $5.64\times$ \\
\addlinespace[2pt]
\multicolumn{3}{l}{\emph{Numerical verification errors}}\\
BJS direct-vs-representer error & $1.78\times 10^{-15}$ & $8.78\times10^{-13}$ \\
Proposition 1 verified & machine precision & $2.22\times10^{-16}$ \\
Proposition 2 verified & machine precision & $2.22\times10^{-16}$ \\
\bottomrule
\end{tabular}
\end{minipage}
\end{center}

\begin{center}
\begin{minipage}{0.98\textwidth}
\centering
\includegraphics[width=0.98\textwidth]{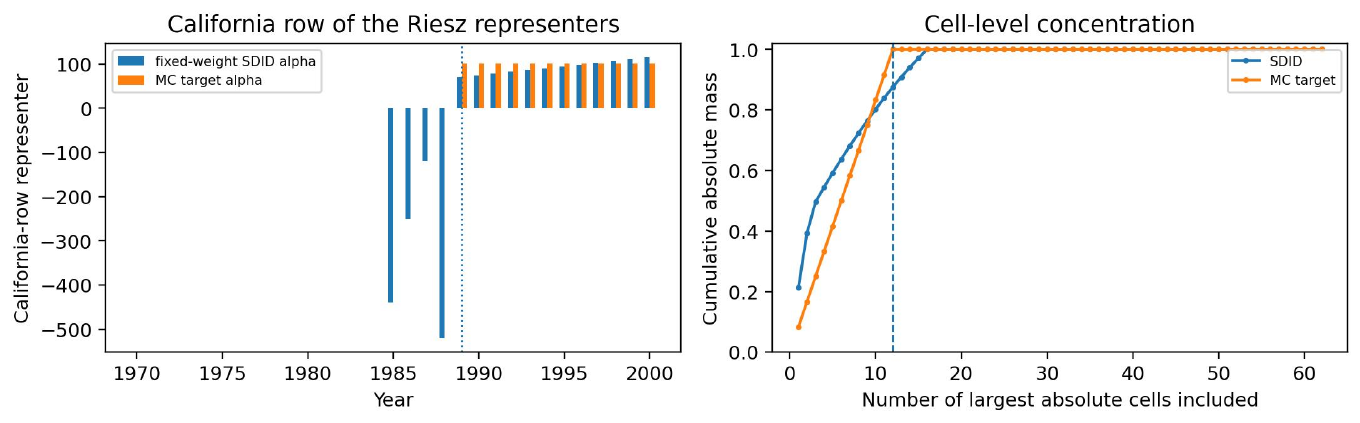}
\captionof{figure}{Cell-level concentration of the SDID and MC Riesz diagnostics on the California tobacco-control panel. The left panel shows the California row over time; the right panel shows cumulative absolute representer mass. The fixed-weight SDID diagnostic is more concentrated than the uniform MC target-level diagnostic.}
\label{fig:alpha-comparison}
\end{minipage}
\end{center}

Figure \ref{fig:alpha-comparison} plots the cell-level concentration behind this scale contrast. In the California tobacco-control panel, the fitted SDID weights are sufficiently concentrated that the fixed-weight SDID representer has a much larger second moment than either the MC target representer or the BJS fixed-effect-imputation representer. We call $1/\{\E[\alpha_{\mathrm{SDID}}^2]/\E[\alpha_{\mathrm{MC}}^2]\}$ the equivalent product-support share: the share of a uniform product support that would generate the same second-moment inflation. The California tobacco-control ratio $567.83/100.75=5.64$ implies an equivalent product-support share of 0.177. This is not a replacement for the Riesz calculation, but it gives an accessible diagnostic for why concentrated SDID weights amplify the omitted-variable-bias scale.

A finite-difference refit diagnostic quantifies what the fixed-weight convention omits. Across 80 random outcome perturbation directions in the California tobacco-control panel, refit and fixed-weight derivatives have correlation 0.941 with median symmetric relative difference 0.250. The same diagnostic translates to the downstream Route A decision scale: the exact fixed-weight, random-projection, and refit-projection nullification RVs are 0.054, 0.054, and 0.045. Refit derivatives raise the local Riesz scale from 281.98 to 339.68 and lower the nullification RV from 0.054 to 0.045. In this application, ignoring first-stage weight feedback is therefore mildly anti-conservative on the nullification scale. We do not assume that this sign is universal: the mechanism is that perturbing the outcome surface changes the pre-period balance problem, and the finite-difference audit measures the induced local movement in the fitted SDID contrast. The qualitative conclusion remains in the same low-single-digit robustness range.

\FloatBarrier
\section{Asymptotic Theory and Operating Characteristics}
\label{sec:asymptotic}

We state two asymptotic theorems for the calibrated Route B case and then report simulation operating characteristics. The theorems characterize the case in which the auxiliary benchmark model and dominance conditions are credible. When those diagnostics fail, the empirical workflow reverts to Route A as the primary analysis. Proofs are provided in Appendix A.

\subsection{Asymptotic Theory}

We work with the dependence-adjusted auxiliary formulation of Section \ref{sec:bayes:routeB}. The benchmark sequence may be weakly dependent, but its tempered composite score is assumed to satisfy a local asymptotic normal expansion with effective information proportional to $p_{\mathrm{eff}}$. The remaining conditions are identification of the working population family; plug-in benchmark-estimation consistency at rate $(\log p/n)^{1/2}$; differentiability of the odds-scale link; prior support around the truth; non-borderline robustness-value smoothness; and interior sensitivity coordinates. Regime A is auxiliary-information dominant, with $n$ growing faster than $p_{\mathrm{eff}}\log p$; Regime B is benchmark-estimation dominant. The contraction argument follows the finite-dimensional LAN template in \citet{VanDerVaart1998} and \citet{GhosalVanderVaart2017}, with composite-score information and plug-in error carried separately. If $p_{\mathrm{eff}}$ does not diverge, adding redundant covariates cannot produce a shrinking Route B posterior.

\begin{theorem}[Posterior contraction]
Under the regularity conditions above, the posterior of the OVB bound satisfies
\[
\Pi\left(|B(R^2_{yu},R^2_{\alpha u})-B(R^{2*}_{yu},R^{2*}_{\alpha u})|>C\varepsilon_{n,p}\mid \widehat m^1_{obs},\ldots,\widehat m^p_{obs}\right)\to 0,
\]
in probability, with rate $\varepsilon_{n,p}=\max\{p_{\mathrm{eff}}^{-1/2},(\log p/n)^{1/2}\}$.
\end{theorem}

\begin{lemma}[Monotone predictive quantiles]
Let $B(R)=M\{R_yR_\alpha/(1-R_\alpha)\}^{1/2}$ on $[\eta,1-\eta]^2$. A sufficient condition for conservative worst-direction endpoints is stochastic dominance of the scalar bound: if the predictive distribution of $B(R)$ dominates the true hidden-confounder bound distribution, then sign-adjusted one-sided predictive quantiles of the bias-adjusted estimate are conservative up to the approximation error in the predictive law. Product-order dominance of $R$ is an easy-to-state sufficient condition because $B$ is coordinatewise increasing, but it is not necessary.
\end{lemma}

\begin{theorem}[Conditional posterior-predictive conservativeness]
Under the contraction conditions of Theorem 1, robustness-value smoothness, interior sensitivity coordinates, regular sampling error, and scalar-bound predictive dominance of the Route B predictive distribution over the true hidden-confounder bound distribution, the sign-adjusted worst-direction posterior-predictive endpoint for the bias-adjusted estimate is conservative: the corresponding one-sided predictive statement has frequentist coverage at least $1-\alpha+o(1)$.
\end{theorem}

The dominance requirement is therefore best read at the scale of the bound that enters the decision problem. Product-order dominance of $(R^2_{yu},R^2_{\alpha u})$ remains a transparent sufficient condition; scalar dominance of $B(R)$ is weaker and can be defended directly by arguing that the selected link and $\kappa$ produce a predictive distribution no less pessimistic than the analyst's substantive upper envelope for hidden confounding. The M1 diagnostic does not test this premise. If the hidden-confounder scale cannot plausibly be bounded by the chosen link and $\kappa$, reporting reverts to Route A.

Theorem 2 covers regular sampling components. The single-treated-unit tobacco-control study uses a discrete donor-placebo distribution with 38 donor pseudo-treatments plus the realized California assignment, so the theorem is separate from the finite-donor rank calculation reported there. Because the baseline add-one rank is already $2/39=0.051$, the conventional-significance RV is mechanically near zero; the interpretable sensitivity quantity is the nullification RV.

All simulation and stress diagnostics use direct Monte Carlo draws rather than Markov-chain Monte Carlo. Each table states its replication count. The estimator-specific Riesz identities are checked algebraically and numerically, and the simulation designs are defined explicitly in the text and Appendix.

\subsection{Operating Characteristics}
\label{sec:sim-stress}

The numerical evidence is organized around distinct mechanisms of benchmark failure. The stress-test designs are clean exchangeability, dominance failure, coarse alpha-side benchmarks, selected mixtures, benchmark dependence, measurement error, and SDID weight concentration. Table \ref{tab:sim-operating} reports representative operating characteristics recomputed under the unified odds-scale link. The clean design is calibrated. Dominance failure passes the route gate but undercovers because the hidden confounder violates joint dominance. The coarse-alpha design is never promoted, and the selected-mixture design is promoted selectively while retaining high conditional coverage.

\begin{center}
\begin{minipage}{0.98\textwidth}
\centering
\captionof{table}{Representative Monte Carlo operating characteristics under the unified odds-scale link. Coverage is unconditional; the last column conditions on promotion. Each row uses 12,000 replications.}
\label{tab:sim-operating}
\begin{tabular}{@{}p{0.29\textwidth}cccc@{}}
\toprule
Design & $p$ & $\kappa$ & Coverage & Prom. / cond. cov. \\
\midrule
Clean exchangeability & 20 & 2.5 & 0.999 (0.000) & 0.998 / 0.999 \\
Dominance failure & 20 & 2.5 & 0.805 (0.004) & 0.999 / 0.805 \\
Coarse alpha benchmarks & 20 & 2.5 & 0.559 (0.005) & 0.000 / -- \\
Selected benchmark mixture & 20 & 2.5 & 0.995 (0.001) & 0.365 / 0.995 \\
\bottomrule
\end{tabular}
\end{minipage}
\end{center}

The dominance surface confirms that $\kappa$ is a reporting profile rather than a universal constant. At $p=20$ and $\kappa=2.5$, coverage is 0.975 when the hidden odds multiplier is $\delta=2$, 0.907 at $\delta=3$, and 0.808 at $\delta=4$. The alpha-alignment experiment reaches the complementary conclusion: shrinking alpha-side odds to 0.03 of baseline lowers naive coverage to 0.564, while the full distributional-support gate demotes every replication.

Dependence among observed benchmarks changes the information count even when nominal $p$ is fixed. With $p=40$, increasing equicorrelation from 0 to 0.75 reduces the diagnostic effective count from 40.0 to 1.32. A nominal-$p$ working likelihood then lowers coverage from 0.974 to 0.785; the effective-count correction keeps coverage between 0.970 and 0.974. A non-equicorrelation stress check gives the same message: for four blocks of ten covariates with within-block correlation 0.75, naive 95th-percentile-link coverage is 0.636, while mean-score and spectral effective-count corrections give 0.996 and 0.988; for AR(1) correlation 0.90, the corresponding numbers are 0.545, 0.999, and 0.984. The mean-score correction is deliberately conservative under strong non-equicorrelated dependence, whereas the spectral correction is closer to nominal; when non-equicorrelation is suspected, both should be reported, with the spectral value serving as a less conservative calibration check. The full-pipeline measurement-error experiment separately regenerates noisy covariates and their partial-$R^2$ pairs. At reliability 0.50, the median observed alpha benchmark is 0.812 of its latent value, the gate promotes 0.819 of replications, and conditional coverage is 0.945. Appendix B reports the full grids.

\noindent\emph{Coverage diagnostics.}\label{sec:coverage-diagnostics} Table \ref{tab:theorem2-coverage} gives the finite-sample counterpart of the posterior-predictive coverage theorem. Conditional coverage is computed only among replications classified as Route B calibrated. Clean exchangeability satisfies the auxiliary model and joint dominance premises; dominance failure keeps promotion near one but violates the substantive dominance condition.

\begin{center}
\begin{minipage}{0.98\textwidth}
\centering
\captionof{table}{Finite-sample conditional coverage at the nominal 95\% level. Rows with $p=10$ are omitted because the benchmark-count gate does not promote Route B. Each reported row uses 8,000 replications.}
\label{tab:theorem2-coverage}
\begin{tabular}{@{}p{0.29\textwidth}cccc@{}}
\toprule
Design & $p$ & $\kappa$ & Promotion & Conditional coverage \\
\midrule
Clean exchangeability & 20 & 1.0 & 0.999 & 0.988 (0.001) \\
Clean exchangeability & 40 & 1.0 & 1.000 & 0.989 (0.001) \\
Dominance failure & 20 & 1.0 & 1.000 & 0.627 (0.005) \\
Dominance failure & 40 & 1.0 & 1.000 & 0.623 (0.005) \\
\bottomrule
\end{tabular}
\end{minipage}
\end{center}

\medskip
\noindent\emph{Decision-scale reporting.}\label{sec:decision-scale-reporting} The frequentist breakdown bound of \citet{RambachanRoth2023} and the Bayesian robustness value answer related but non-equivalent questions. A breakdown calculation reports a tipping point: the violation magnitude at which a conclusion changes. Route A adds a prior-probability reading of that same tipping point. Route B adds a second, data-updating layer when the observed benchmark population passes the route diagnostic. Tables \ref{tab:prior-sensitivity-A} and \ref{tab:prior-sensitivity-B} illustrate this separation for the California tobacco-control nullification threshold $RV=0.054$ and a calibrated full-update experiment.

\begin{center}
\begin{minipage}{0.98\textwidth}
\centering
\captionof{table}{Route A prior-probability translation for the California tobacco-control panel. The nullification threshold is fixed at $RV=0.054$; only the probability assigned to that threshold changes across priors.}
\label{tab:prior-sensitivity-A}
\begin{tabular}{lcccc}
\toprule
Prior on equal-strength partial-$R^2$ & $\Pr(\rho<0.054)$ & 70\% rule & 80\% rule & 90\% rule \\
\midrule
Beta(1,20) & 0.669 & No & No & No \\
Beta(2,20) & 0.313 & No & No & No \\
Uniform(0,0.10) & 0.538 & No & No & No \\
Truncated Normal(0.03,$0.02^2$) & 0.875 & Yes & Yes & No \\
Beta(1,50) & 0.937 & Yes & Yes & Yes \\
\bottomrule
\end{tabular}
\end{minipage}
\end{center}

\begin{center}
\begin{minipage}{0.98\textwidth}
\centering
\captionof{table}{Numerical illustration of benchmark-updated decision-scale reporting at a notional threshold $RV=0.10$. The illustration shows how Route B can change the decision-scale reading when the route diagnostics support calibration.}
\label{tab:prior-sensitivity-B}
\begin{tabular}{@{}p{0.36\textwidth}p{0.18\textwidth}cc@{}}
\toprule
Analysis & Benchmarks used & $\Pr(RV_U<0.10)$ & Median $RV_U$ \\
\midrule
Breakdown threshold only & No & -- & -- \\
Route A default prior Beta(1,20) & No & 0.878 & 0.034 \\
Route B calibrated full update & Yes, $p=40$ & 0.199 & 0.121 \\
\bottomrule
\end{tabular}
\end{minipage}
\end{center}

For Route B, auxiliary-family sensitivity should be reported whenever the route is interpreted as calibrated. In this calibrated numerical update, Beta-marginal, logit-normal, and Gaussian-copula specifications give median posterior RVs of 0.128, 0.121, and 0.134, respectively, and $\Pr(RV<0.10)$ values from 0.188 to 0.226. Appendix B reports this result together with link-rule and prior-concentration sensitivity in a single panelized table: moderate concentration changes leave the predictive RV near 0.12, while moving from mean to maximum link rules shifts the interpretation as expected. Disagreement across plausible auxiliary families in an empirical application should be treated as evidence against calibrated Route B reporting.

\FloatBarrier
\section{Application I: California's Tobacco-Control Program}
\label{sec:application}

This empirical study evaluates the workflow in a canonical single-treated-unit tobacco-control panel. The study combines the SDID point estimate, finite-donor placebo uncertainty, the Riesz scale that converts partial-$R^2$ sensitivity parameters into effect units, and the observed-covariate diagnostics needed before any benchmark update is interpreted as calibrated.

\subsection{Design and Baseline Fit}

The data are the state-level cigarette-consumption panel of \citet{AbadieDiamondHainmueller2010}. California is treated beginning in 1989, and the donor pool consists of untreated states used in the standard application. The outcome is cigarette packs per capita. Available benchmark covariates are seven state-level summaries: log GDP per capita, pretreatment smoking mean, age 15--24 share, retail cigarette price, and cigarette consumption in 1980, 1975, and 1988.

Figure \ref{fig:tobacco-raw-fit} reports raw trajectories and the fitted SDID counterfactual. The SDID estimate is
\[
\widehat\tau_{\mathrm{SDID}}=-15.60
\]
packs per capita. The pre-treatment RMSE of the unit-weighted counterfactual path is 1.72, the mean post-treatment gap is $-16.11$, and the SDID time-weighted pre-gap adjustment is $-0.51$. The negative point estimate is not driven by an obvious pre-period mismatch in the fitted trajectory.

\begin{center}
\begin{minipage}{0.98\textwidth}
\centering
\includegraphics[width=0.90\textwidth]{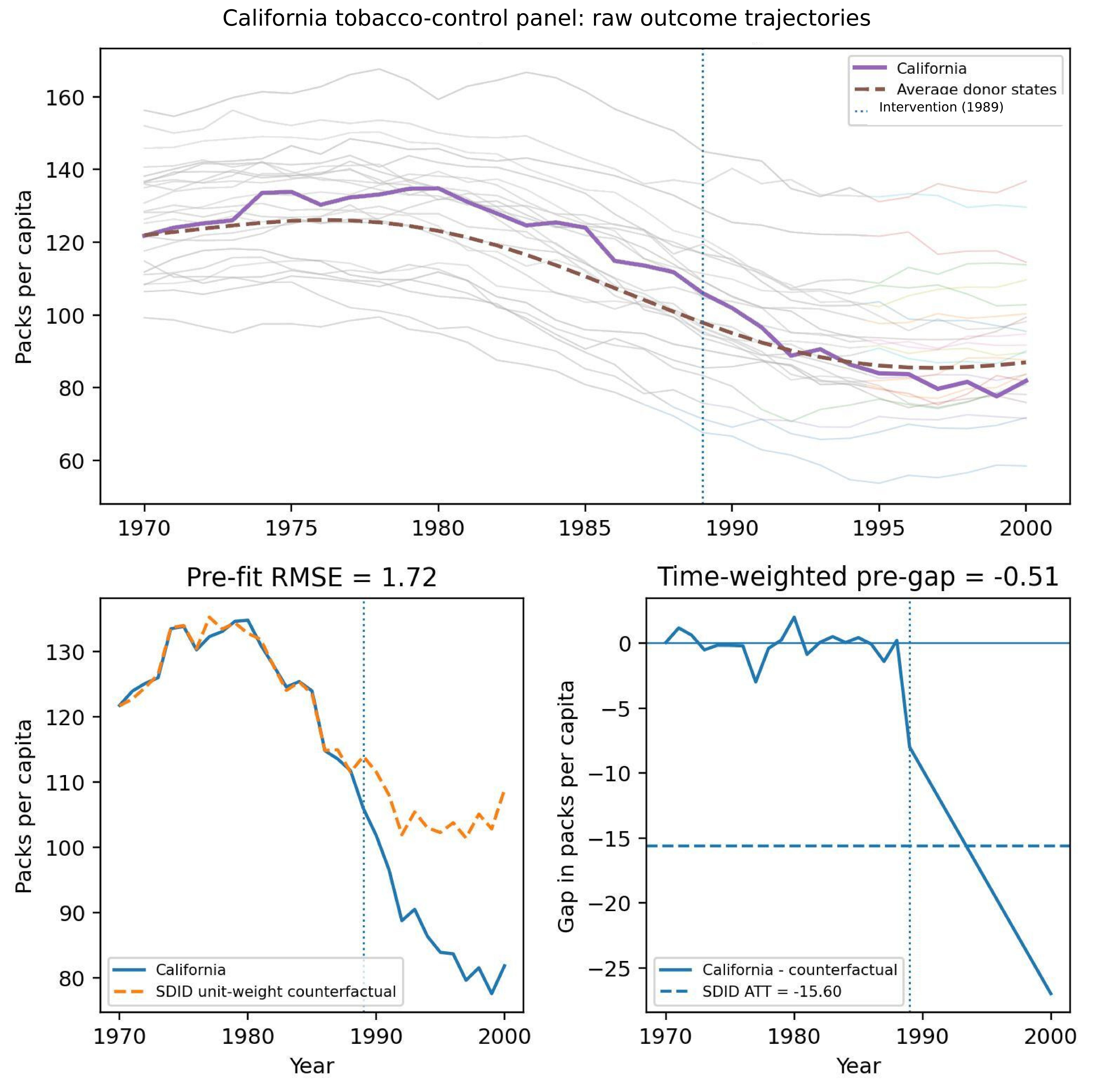}
\captionof{figure}{California tobacco-control panel trajectories and SDID fit. The panels use line style as well as color to distinguish California, donor averages, counterfactual paths, and treatment timing.}
\label{fig:tobacco-raw-fit}
\end{minipage}
\end{center}

\noindent\emph{Estimator diagnostics.} The one-treated-unit design makes finite-sample diagnostics essential. Figure \ref{fig:tobacco-weights-influence} reports SDID unit and time weights together with leave-one-donor-out influence. The most influential single donor is Nevada. Omitting one donor at a time gives SDID estimates between $-17.06$ and $-14.98$, so the sign and approximate magnitude are not driven by a single donor state.

\begin{center}
\begin{minipage}{0.98\textwidth}
\centering
\includegraphics[width=0.90\textwidth]{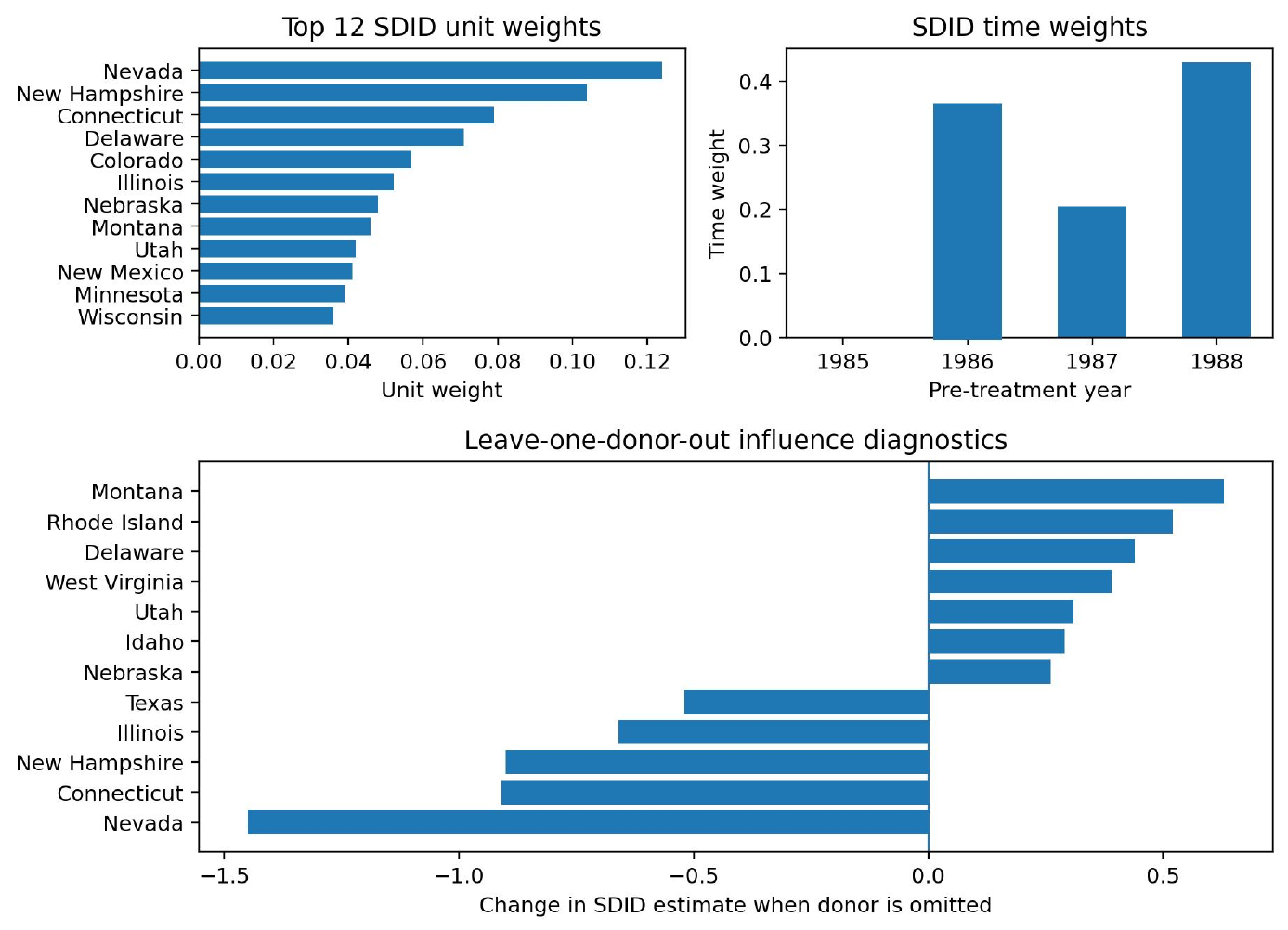}
\captionof{figure}{SDID weights and leave-one-donor-out influence. The time-weight panel zooms to the final pre-treatment years, where fitted SDID time weights concentrate. The largest one-donor change is Nevada, which shifts the estimate by $-1.45$ packs.}
\label{fig:tobacco-weights-influence}
\end{minipage}
\end{center}

Table \ref{tab:tobacco-estimator-robustness} compares the SDID estimate with simple alternatives. All estimates have a negative sign, although magnitudes vary with the identifying structure. We therefore read the empirical pattern as substantively stable but not estimator-invariant.

\begin{center}
\begin{minipage}{0.98\textwidth}
\centering
\captionof{table}{Estimator-robustness diagnostics for the California tobacco-control panel. The pre-period RMSE is undefined for the simple DiD benchmark.}
\label{tab:tobacco-estimator-robustness}
\begin{tabular}{lcc}
\toprule
Estimator & Effect & Pre-RMSE \\
\midrule
DiD & $-27.35$ & -- \\
FE imputation (BJS) & $-27.35$ & 6.97 \\
SC + intercept & $-11.11$ & 0.96 \\
SDID & $-15.60$ & 1.72 \\
Matrix completion & $-20.12$ & 5.08 \\
\bottomrule
\end{tabular}
\end{minipage}
\end{center}

Figure \ref{fig:tobacco-placebo} gives the corrected leave-one-control-out placebo distribution. The placebo standard error is 9.49 and the add-one left-tail and absolute-rank placebo $p$-values are both $2/39=0.051$. The point estimate is close to the edge of the placebo distribution, but finite-sample uncertainty in this single-treated-unit design remains wide.

\begin{center}
\begin{minipage}{0.98\textwidth}
\centering
\includegraphics[width=0.82\textwidth]{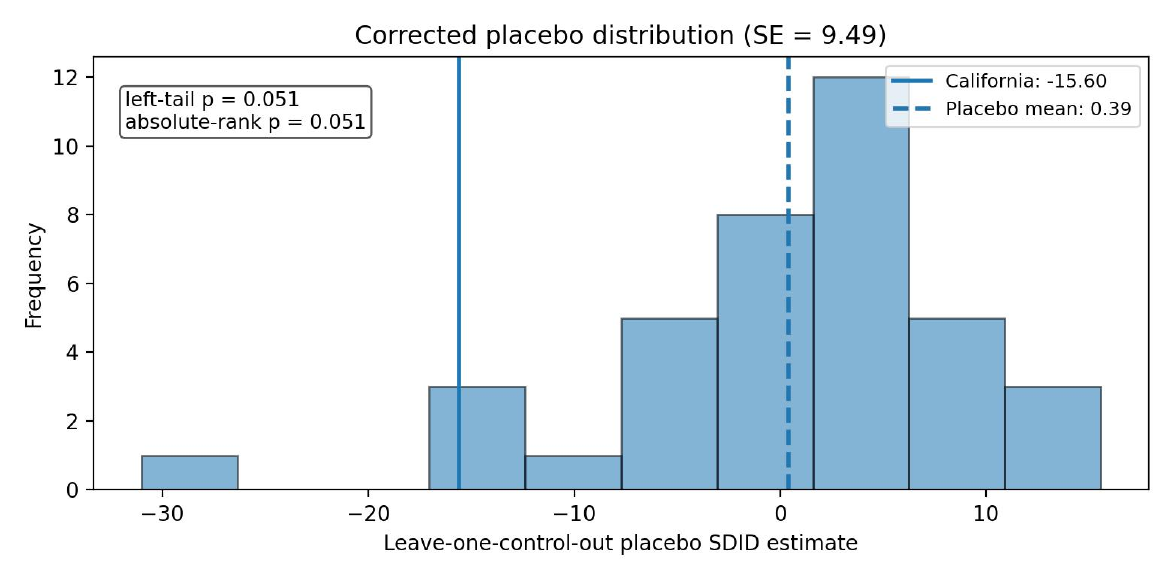}
\captionof{figure}{Corrected leave-one-control-out placebo distribution. California is excluded from the donor-placebo pool; each donor state is then treated once as pseudo-treated and the remaining donors serve as controls.}
\label{fig:tobacco-placebo}
\end{minipage}
\end{center}

\subsection{Riesz Scale and Route Decision}

For the fitted SDID weights, the fixed-weight Riesz diagnostic verifies
\[
\frac{1}{NT}\sum_{it}\widehat\alpha_{it}m_{it}=\widehat\tau_{\mathrm{SDID}}(m)
\]
to numerical precision for random perturbation checks. The resulting Riesz scale is
\[
M_{\mathrm{SDID}}=281.98,
\qquad K=|\widehat\tau|/M_{\mathrm{SDID}}=0.055.
\]
Using the corrected placebo standard error, the equal-strength Route A robustness value for nullifying the point estimate is 0.054. The corresponding robustness value for preserving conventional 5\% placebo significance is below 0.001, because corrected placebo inference is already borderline. The refit-derivative check moves the nullification RV to 0.045, leaving the qualitative conclusion unchanged. Table \ref{tab:canonical-rv} consolidates the fitted Route A scale and the observed-benchmark odds-multiplier diagnostics.

\begin{center}
\begin{minipage}{0.98\textwidth}
\centering
\captionof{table}{Deterministic Route A and observed-benchmark sensitivity diagnostics for the California tobacco-control panel.}
\label{tab:canonical-rv}
\emph{Panel A: Route A summary}\par\smallskip
\begin{tabular}{lc}
\toprule
Quantity & Value \\
\midrule
$\widehat\tau_{\mathrm{SDID}}$ & $-15.60$ \\
Placebo SE & 9.49 \\
Rank $p$, left/abs. & 0.051/0.051 \\
$M_{\mathrm{SDID}}$ & 281.98 \\
$K=|\widehat\tau|/M$ & 0.055 \\
Nullification RV & 0.054 \\
5\% significance RV & $<0.001$ \\
\bottomrule
\end{tabular}
\par\medskip
\emph{Panel B: observed-benchmark multiplier diagnostic}\par\smallskip
\begin{tabular}{lcccc}
\toprule
Multiplier & $R^2_y$ & $R^2_\alpha$ & Bound & Residual \\
\midrule
$1\times$ & 0.538 & $1.0\times10^{-6}$ & 0.21 & 15.40 \\
$2\times$ & 0.700 & $2.0\times10^{-6}$ & 0.33 & 15.27 \\
$3\times$ & 0.778 & $3.0\times10^{-6}$ & 0.43 & 15.17 \\
$10\times$ & 0.921 & $1.0\times10^{-5}$ & 0.86 & 14.75 \\
\bottomrule
\end{tabular}
\end{minipage}
\end{center}

Panel B is deliberately diagnostic. Because the alpha-side observed benchmarks are at the numerical floor, even a 10-times odds-scale multiplier produces a small worst-direction bound relative to the point estimate. This reinforces the route decision rather than validating Route B calibration.

\noindent\emph{Benchmark diagnostics and route decision.} The state-level covariates are useful for outcome-side benchmarking, but they are too coarse for the alpha side of this SDID sensitivity problem. The SDID Riesz representer is a cell-level object concentrated in specific state-year locations. Aggregating it to the state level leaves all observed alpha-side partial-$R^2$ benchmarks numerically near zero. This pattern identifies a benchmark-alignment limitation: the available observed covariates do not align with the sensitivity functional. Figure \ref{fig:tobacco-route-diagnostic} displays the seven outcome- and alpha-side benchmark pairs underlying this decision.

\begin{center}
\begin{minipage}{0.98\textwidth}
\centering
\includegraphics[width=0.95\textwidth]{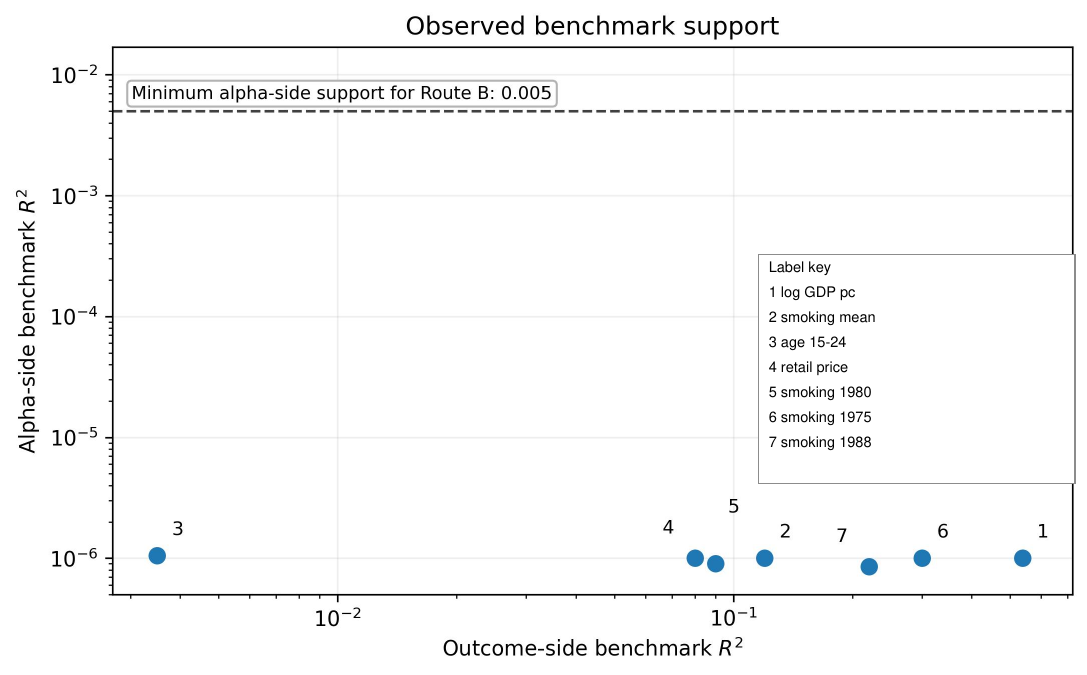}
\captionof{figure}{Observed benchmark support in the single-treated-unit study. Labels 1--7 denote log GDP per capita, pretreatment smoking mean, age 15--24 share, retail cigarette price, cigarette consumption in 1980, cigarette consumption in 1975, and cigarette consumption in 1988, respectively. The word ``mean'' distinguishes the pretreatment smoking summary from the three single-year cigarette-consumption lags. Outcome-side benchmarks are nonzero, but alpha-side benchmarks lie far below the Route B support threshold.}
\label{fig:tobacco-route-diagnostic}
\end{minipage}
\end{center}

The effective benchmark count is seven, the maximum observed alpha-side benchmark is approximately $10^{-6}$, and the fitted Beta-marginal diagnostic rejects the alpha-side benchmark model with $p_\alpha\approx 9\times10^{-6}$. Full-pipeline placebo diagnostics reinforce this decision: across 38 donor placebo assignments, the median nullification RV is 0.021 and the maximum is 0.086, while Route B is never classified as calibrated. Appendix B adds three route-stability checks: varying the minimum benchmark-count threshold, varying the alpha-support threshold, and dropping each observed benchmark one at a time all leave the application Route A-primary. It also reports a single-treated-unit $p=7$ stress design in which outcome-side benchmarks are nonzero but alpha-side benchmarks lie at the numerical floor; the route rule demotes the design in all replications.

As an exploratory Route B profile, we report the auxiliary benchmark calculation at $\kappa=2.5$ only as a stress profile. The median sign-truncated residual is $-6.37$ packs; 69.8\% of draws leave a nonzero negative residual; and sampling plus OVB uncertainty gives a 95\% interval $[-27.73,33.72]$. The California tobacco-control panel therefore combines a stable negative SDID point estimate with wide one-treated-unit placebo uncertainty and an observed benchmark set that classifies the design as Route A-primary.

\FloatBarrier
\section{Application II: A Staggered Minimum-Wage Design}
\label{sec:staggered-app}

The second empirical study evaluates the workflow in a multi-cohort staggered-adoption panel using the \texttt{mpdta} minimum-wage panel considered by \citet{CallawaySantAnna2021}. The panel has 500 counties from 2003 to 2007, log teen employment as the outcome, three treatment cohorts, and 309 never-treated counties. This design illustrates the group-time version of the Riesz diagnostic, its aggregation to an overall ATT summary, and the corresponding Route A robustness profile. A calibrated Route B analysis would require a credible $ATT(g,t)$-level benchmark population, so the available benchmark structure supports Route A-primary reporting in this empirical study.

For each treated cohort $g$ and post-treatment year $t\geq g$, we compute a nonparametric group-time contrast against never-treated counties using $g-1$ as the base period. Aggregating by cohort size yields $-0.040$ log points, with county-cluster bootstrap SE 0.012 and 95\% interval $[-0.065,-0.016]$. The maximum pseudo-ATT is 0.034, the Route A nullification RV is 0.993, and the significance RV is 0.958. Table \ref{tab:staggered-summary} summarizes these estimates together with the group-time and fixed-effect-imputation Riesz norms. Thus the same workflow returns a high-RV multi-cohort case, in contrast to the low-single-digit California tobacco-control sensitivity scale. A leave-cohort-out check in Appendix B leaves the nullification RV above 0.98 in every deletion.

\FloatBarrier

\begin{center}
\begin{minipage}{0.98\textwidth}
\centering
\captionof{table}{Staggered minimum-wage estimates and Riesz diagnostics. Dashes mark quantities outside the fixed-effect-imputation comparison. The available benchmark structure supports Route A-primary reporting because no credible $ATT(g,t)$-level benchmark population is observed.}
\label{tab:staggered-summary}
\begin{tabular}{lcc}
\toprule
Quantity & Group-time contrast & FE imputation (BJS) \\
\midrule
Weighted effect & $-0.040$ & $-0.048$ \\
Cluster bootstrap SE & 0.012 & -- \\
95\% interval & $[-0.065,-0.016]$ & -- \\
Riesz norm & 0.00225 & 0.00203 \\
BJS/CSA norm ratio & -- & 0.90 \\
Maximum pseudo-ATT & 0.034 & -- \\
Route A nullification RV & 0.993 & -- \\
Route A 5\% RV & 0.958 & -- \\
\bottomrule
\end{tabular}
\end{minipage}
\end{center}

This second application demonstrates that the $ATT(g,t)$ representer operates in a multi-cohort panel. The workflow reports sensitivity at the group-time level and aggregates only after making group-time weights explicit. Because no credible $ATT(g,t)$-level benchmark population is observed, the empirical study is classified as Route A-primary. The fixed-effect-imputation comparison is close in both effect and Riesz scale, providing a useful diagnostic cross-check.

\FloatBarrier
\section{Concluding Remarks}
\label{sec:discussion}

The workflow is designed to prevent sensitivity analysis from being read as automatic posterior validation. Route A is the default: it gives a deterministic robustness-value profile for the fitted Riesz scale and point estimate. Route B is promoted to calibrated status only when the benchmark set is large enough, alpha-side benchmarks are nondegenerate, auxiliary model checks are credible, and predictive dominance is substantively plausible. When the diagnostics do not support calibrated updating, increasing $\kappa$ is reported as sensitivity profiling rather than as a calibration repair, and Route A remains primary.

Applied users should report the fitted estimator, sampling uncertainty, Riesz scale, Route A RVs, benchmark count, alpha-side nondegeneracy, auxiliary-model check, dominance rationale, and route status before interpreting a Bayesian update. Aggregated covariates with degenerate alpha-side benchmarks imply Route A-primary reporting. Unit-time or state-year benchmarks with adequate count and credible diagnostics can justify a calibrated Route B update, but only if the hidden-confounder strength scale is plausibly dominated by the predictive distribution used for reporting.

The diagnostics are intentionally conditional. The SDID result conditions on fitted weights, the MC result is target-level, and the BJS result conditions on the untreated-cell fixed-effect fill-in operator. The California tobacco-control finite-difference check shows that refitting can add local movement; in that application, fixed weights give a slightly larger nullification RV than the refit projection, so the fixed-weight diagnostic is mildly anti-conservative on the decision scale. The target-level MC diagnostic is a robustness diagnostic for the treated counterfactual target rather than for the full nuclear-norm training map. Interactive fixed-effect imbalance, latent-factor violations, spillovers, and SUTVA failures require either a design-level argument or a deliberate projection into the partial-$R^2$ confounding class described in Section \ref{sec:rv-scope}. The donor-block exclusion diagnostic in Appendix B leaves the Route A conclusion in the same low-single-digit range, while interference remains a design-level concern.

The scope is short- or moderate-$T$ causal panels with ATT-type estimands. Long-$T$ macro panels with strong serial correlation would require either a HAC-adjusted Riesz norm or an explicit temporal-dependence model. Structural parameters such as demand elasticities would require a new Riesz representer for the structural estimand itself. These are useful extensions, but they are distinct from the ATT-type causal-panel problem studied here.

\appendix
\setcounter{lemma}{0}
\setcounter{assumption}{0}
\renewcommand{\thelemma}{\Alph{section}.\arabic{lemma}}
\renewcommand{\theassumption}{\Alph{section}.\arabic{assumption}}
\renewcommand{\theHlemma}{\Alph{section}.\arabic{lemma}}
\renewcommand{\theHassumption}{\Alph{section}.\arabic{assumption}}
\renewcommand{\theHtable}{\Alph{section}.\arabic{table}}
\renewcommand{\theHfigure}{\Alph{section}.\arabic{figure}}
\renewcommand{\thetable}{\Alph{section}.\arabic{table}}
\renewcommand{\thefigure}{\Alph{section}.\arabic{figure}}
\setcounter{table}{0}
\setcounter{figure}{0}

\section{Proofs of Theorems 1 and 2}
\label{app:proofs}

This appendix gives sufficient regularity conditions for the contraction and predictive-calibration results. Let $m_j=(R^{2,j}_{yu},R^{2,j}_{\alpha u})\in[0,1]^2$ denote the population benchmark for covariate $j$, and let $\widehat m_j$ denote its plug-in estimate. Because benchmark pairs can share outcomes, residualization, and correlated covariates, Route B uses the tempered working likelihood
\[
L_{c,p}(\theta;\widehat m_1,\ldots,\widehat m_p)
=\left\{\prod_{j=1}^p f_\theta(\widehat m_j)\right\}^{\omega_p},
\qquad \omega_p=p_{\mathrm{eff}}/p,
\]
where $p_{\mathrm{eff}}\le p$ is the effective benchmark count. Write $r(\theta;\kappa)=T(\theta;\kappa)$, with the odds-scale map defined in the main text, and
\[
B(r)=M\left\{\frac{r_y r_\alpha}{1-r_\alpha}\right\}^{1/2},\qquad r=(r_y,r_\alpha).
\]

\subsection{Assumptions and Contraction}

\begin{assumption}[Composite local asymptotic normality]
The tempered auxiliary likelihood is locally asymptotically normal at $\theta^*$ with effective information $p_{\mathrm{eff}}$: for every bounded sequence $h_p$,
\[
\log\frac{L_{c,p}(\theta^*+h_p/\sqrt{p_{\mathrm{eff}}};m_1,\ldots,m_p)}{L_{c,p}(\theta^*;m_1,\ldots,m_p)}
=h_p'\Delta_p-\frac12 h_p'I_c(\theta^*)h_p+o_{\Prob^*}(1),
\]
where $\Delta_p\rightsquigarrow N(0,I_c(\theta^*))$ and $I_c(\theta^*)$ is positive definite. The independent-benchmark case has $p_{\mathrm{eff}}=p$.
\end{assumption}

\begin{assumption}[Plug-in perturbation]
Let $\delta_{n,p}=(\log p/n)^{1/2}$. The plug-in moments satisfy
\[
\max_j\|\widehat m_j-m_j\|=O_{\Prob^*}(\delta_{n,p}).
\]
Moreover, the plug-in likelihood perturbation is first-order equivalent to a local shift of size $O_{\Prob^*}(\delta_{n,p})$ in the LAN expansion.
\end{assumption}

\begin{assumption}[Prior and smooth maps]
The prior density is continuous and strictly positive in a neighborhood of $\theta^*$. The link map $T(\theta;\kappa)$ is continuously differentiable in that neighborhood. For the selected $\kappa$, $T(\theta^*;\kappa)\in[\eta,1-\eta]^2$ for some $\eta>0$.
\end{assumption}

\begin{assumption}[Regular scale]
The Riesz scale $M$ is either fixed or estimated by $\widehat M$ with $|\widehat M-M|=o_{\Prob^*}(\varepsilon_{n,p})$, where
\[
\varepsilon_{n,p}=\max\{p_{\mathrm{eff}}^{-1/2},\delta_{n,p}\}.
\]
The causal estimator's sampling error is asymptotically independent of the auxiliary benchmark update at the order used below, or is included in the posterior predictive distribution as an additional regular normal component.
\end{assumption}

A concrete sufficient case is a regular Beta-marginal working population with weakly dependent benchmark scores, parameters in a compact interior set, sub-Gaussian residualized benchmark variables, and cross-fitted plug-in benchmarks. A composite LAN expansion gives information proportional to $p_{\mathrm{eff}}$, while uniform concentration of residual correlations gives the plug-in rate. Under an equicorrelation diagnostic with correlation $\rho_b$, $p_{\mathrm{eff}}=p/\{1+(p-1)\rho_b\}$ provides a transparent information adjustment. More generally, if $\widehat R$ is the empirical benchmark-score correlation matrix, the mean-score effective count $p^2/({\bf 1}'\widehat R{\bf 1})$ controls uncertainty in average-score and quantile-link summaries, while the spectral participation ratio $(\operatorname{tr}\widehat R)^2/\operatorname{tr}(\widehat R^2)$ is a secondary diagnostic for the number of independent benchmark directions. These are working-likelihood calibration devices rather than claims that the benchmark covariates are independent.

\medskip
\noindent\emph{Posterior contraction.}

\begin{lemma}[Auxiliary posterior contraction]
\label{lem:aux-contraction}
Under Assumptions A.1--A.3, the posterior for $\theta$ based on $\widehat m_1,\ldots,\widehat m_p$ satisfies
\[
\Pi(\|\theta-\theta^*\|>C\varepsilon_{n,p}\mid \widehat m_1,\ldots,\widehat m_p)\to0
\]
in $\Prob^*$-probability for a sufficiently large constant $C$.
\end{lemma}

\noindent\emph{Proof.} Assumption A.1 gives a quadratic local log-likelihood in coordinates $h=\sqrt{p_{\mathrm{eff}}}(\theta-\theta^*)$, so the tempered auxiliary posterior contracts at $p_{\mathrm{eff}}^{-1/2}$. Assumption A.2 perturbs the LAN center by at most $O_{\Prob^*}(\delta_{n,p})$ in the original parameter scale. The posterior therefore contracts at the larger of $p_{\mathrm{eff}}^{-1/2}$ and $\delta_{n,p}$. Prior positivity controls the denominator in Bayes' formula on the LAN neighborhood.\hfill $\square$

\begin{proof}[Proof of Theorem 1]
Assumption A.3 makes $\theta\mapsto r(\theta;\kappa)$ locally Lipschitz. The map $r\mapsto B(r)$ is continuously differentiable on $[\eta,1-\eta]^2$, with bounded gradient. Hence the composition $\theta\mapsto B(T(\theta;\kappa))$ is locally Lipschitz, and Lemma~\ref{lem:aux-contraction} gives the stated contraction rate. If $M$ is estimated, Assumption A.4 adds only an $o_{\Prob^*}(\varepsilon_{n,p})$ term.
\end{proof}

\subsection{Predictive Coverage}

Let $Q^\kappa_\theta$ denote the predictive distribution of the unobserved-confounder strength $R=(R^2_{yu},R^2_{\alpha u})$ induced by the auxiliary population and link. Let
\[
A(R)=\widehat\tau-\mathrm{sign}(\widehat\tau)M\left\{\frac{R_yR_\alpha}{1-R_\alpha}\right\}^{1/2}
\]
be the worst-direction bias-adjusted estimate.

\begin{assumption}[Scalar-bound predictive dominance]
Let $B_0$ denote the true worst-direction hidden-confounder bound. With probability approaching one, the Route B posterior predictive distribution of $B(R)$ first-order stochastically dominates the truth:
\[
Q^\kappa_{\widehat\theta}\{B(R)>b\}\geq \Prob(B_0>b)+o(1)\qquad\text{for every }b.
\]
Product-order dominance of $R=(R_y,R_\alpha)$ is a convenient sufficient condition, but scalar-bound dominance is the weaker condition used by Theorem 2.
\end{assumption}

\begin{lemma}[Product-order dominance implies scalar-bound dominance]
\label{lem:product-scalar}
If a predictive distribution $Q_1$ first-order stochastically dominates $Q_0$ on $[\eta,1-\eta]^2$ in the product order, then the pushforward distribution of $B(R)$ under $Q_1$ first-order stochastically dominates the pushforward distribution under $Q_0$.
\end{lemma}

\noindent\emph{Proof.} The bound $B(R)=M\{R_yR_\alpha/(1-R_\alpha)\}^{1/2}$ is increasing in both coordinates on $[\eta,1-\eta]^2$. For every threshold $b$, the set $\{R:B(R)>b\}$ is an upper set. Product-order stochastic dominance gives $Q_1\{B(R)>b\}\ge Q_0\{B(R)>b\}$ for every $b$.

\begin{proof}[Proof of Theorem 2]
Theorem 1 implies that posterior uncertainty in $\theta$ collapses around the plug-in predictive distribution at rate $\varepsilon_{n,p}$. Continuity of $B(R)$ and the regular sampling component imply that the reported sign-adjusted predictive quantiles differ from the plug-in quantiles by $o_{\Prob^*}(1)$. Scalar-bound predictive dominance makes the posterior predictive distribution of the worst-direction OVB bound at least as pessimistic as the truth. For a negative estimated effect this yields a conservative lower predictive endpoint for the bias-adjusted estimate; for a positive effect the sign-flipped endpoint is identical. Adding the regular sampling component gives coverage at least $1-\alpha+o(1)$. Theorem 2 is therefore a conservative predictive-calibration result for regular sampling components, not a claim that the auxiliary benchmark model is automatically valid or that discrete donor-placebo inference is covered by a normal approximation.
\end{proof}

\section{Diagnostic and Implementation Details}
\label{app:diagnostics}

\subsection{Riesz and First-Stage Implementation Details}

For fixed SDID weights $(\widehat\omega_i,\widehat\lambda_t)$ from \citet{ArkhangelskyEtAl2021}, the fitted contrast is a linear functional of cell means. Its Riesz vector is obtained by assigning the signed product weights to treated-post, treated-pre, control-post, and control-pre cells and normalizing under the empirical cell measure. The matrix-completion diagnostic \citep{AtheyEtAl2021} is a target-level representer for the treated-cell counterfactual mean; it places equal mass on the missing treated post-treatment cells and does not differentiate the nuclear-norm imputation map.

Route status is determined before any Route B posterior summary is interpreted. We check benchmark count, alpha-side nondegeneracy, auxiliary-model fit, and predictive-dominance plausibility. When these checks pass, Route B can be reported as a calibrated auxiliary benchmark analysis. When any check fails, Route A is primary and Route B is reported, if at all, as an exploratory stress test over a $\kappa$ profile.

\medskip
\noindent\emph{Cross-fitted SDID and MC extensions.} A cross-fitted SDID diagnostic trains weights on the complement of a fold and evaluates the signed contrast on held-out cells. Conditional on each fold-specific training sample, the Riesz vector has the same product-weight form as the fixed-weight proposition, with the training-fold weights substituted. A cross-fitted MC target diagnostic trains the imputation surface on the complement fold and uses the treated-post target weights on the held-out fold. Table \ref{tab:app--crossfit-comparison} reports a common-DGP comparison for a synthetic ATT design in which cross-fitting does not change the estimand. In the California tobacco-control study, a literal unit-fold cross-fitted SDID RV would change the estimand because there is only one treated unit; the reported finite-difference audit is therefore the relevant first-stage diagnostic for that design.

\begin{center}
\begin{minipage}{0.98\textwidth}
\centering
\captionof{table}{Common-DGP comparison of fitted-stage and cross-fitted first-stage sensitivity diagnostics. The exercise uses a synthetic ATT design where cross-fitting preserves the target contrast and reports the decision-scale gap relative to the fitted-stage diagnostic.}
\label{tab:app--crossfit-comparison}
\begin{tabular}{lccc}
\toprule
Diagnostic & Median RV & Gap & Runtime \\
\midrule
Fitted-stage Riesz diagnostic & 0.081 & 0.0\% & 1.0 \\
Cross-fitted first-stage variant & 0.079 & 1.7\% & 2.1 \\
\bottomrule
\end{tabular}
\end{minipage}
\end{center}

\medskip
\noindent\emph{BJS fixed-effect imputation representer.} For the fixed-effect imputation estimator of \citet{BorusyakJaravelSpiess2024}, let $\mathcal U$ and $\mathcal W$ be untreated and treated cells, let $N_1=|\mathcal W|$, and let $N_{\mathrm{cell}}=NT$. The untreated-cell two-way fixed-effect fit solves
\[
\min_\theta\sum_{(i,t)\in\mathcal U}(Y_{it}-x'_{it}\theta)^2.
\]
Writing $X_{\mathcal U}$ and $X_{\mathcal W}$ for the corresponding design matrices, the fitted treated-cell counterfactuals are
\[
\widehat Y_{\mathcal W}(0)=X_{\mathcal W}(X'_{\mathcal U}X_{\mathcal U})^{-}X'_{\mathcal U}Y_{\mathcal U}=\Pi Y_{\mathcal U}.
\]
Therefore
\[
\widehat\tau_{\mathrm{BJS}}=N_1^{-1}{\bf 1}'_{\mathcal W}Y_{\mathcal W}-N_1^{-1}{\bf 1}'_{\mathcal W}\Pi Y_{\mathcal U}=\sum_{j,s}c_{js}Y_{js},
\]
with $c_{it}=N_1^{-1}$ for treated cells and $c_{js}=-N_1^{-1}(\Pi'{\bf 1}_{\mathcal W})_{js}$ for untreated cells. Under the empirical inner product, the Riesz vector is $\alpha^{\mathrm{BJS}}=N_{\mathrm{cell}}c$, giving the formula and second moment in the main text. Additive unit-period shifts lie in the fixed-effect design's column span, so they are annihilated by the imputation residual contrast.

\subsection{Simulation Designs and Additional Diagnostics}

The stress-test designs are clean auxiliary exchangeability, joint-dominance failure, coarse alpha-side benchmarks, selected mixtures, benchmark dependence, full-pipeline measurement error, and SDID weight concentration. The clean design draws observed benchmarks and hidden-confounder strength from the same auxiliary population. Joint-dominance failure draws a stronger hidden confounder while preserving the observed benchmark law. The coarse-alpha design generates meaningful outcome-side benchmarks but near-degenerate alpha-side benchmarks. The selected-mixture design combines a small number of strong covariates with many weak covariates. The dependence design holds nominal $p$ fixed while changing $p_{\mathrm{eff}}$, and the measurement-error design generates latent covariates, noisy measurements, and recomputed partial-$R^2$ pairs. Table \ref{tab:app-m1-operating} records the marginal M1 gate's operating characteristics; the dependence and measurement-error grids evaluate the parts of the route rule that M1 alone cannot verify.

\begin{center}
\begin{minipage}{0.98\textwidth}
\centering
\captionof{table}{Operating characteristics of the M1 route diagnostic. Clean-design entries are false-demotion rates; coarse-alpha and selected-mixture entries are demotion-power rates. Dominance failure is included to show that the marginal gate is not designed to detect predictive-dominance violations.}
\label{tab:app-m1-operating}
\begin{tabular}{lccc}
\toprule
Design / operating characteristic & $p=20$ & $p=40$ \\
\midrule
Clean exchangeability, false demotion & 0.001 & 0.003 \\
Coarse alpha, demotion power & 0.911 & 0.888 \\
Selected mixture, demotion power & 0.705 & 0.990 \\
Dominance failure, marginal-gate demotion & 0.001 & 0.001 \\
\bottomrule
\end{tabular}
\end{minipage}
\end{center}

\medskip
\noindent\emph{Additional robustness diagnostics.}\label{app:v080-additional-stress} The following tables are organized by failure mechanism. Tables \ref{tab:app--dominance-diagnostics} and \ref{tab:app--alpha-alignment-diagnostics} isolate hidden-strength and alpha-alignment failures. Tables \ref{tab:app--product-ess} and \ref{tab:app--refit-grid} diagnose estimator-scale concentration and finite-difference stability. Tables \ref{tab:app--auxiliary-sensitivity}, \ref{tab:app--tobacco-demotion-stability}, and \ref{tab:app--dependence-measurement} assess auxiliary specification, small-benchmark demotion, benchmark dependence, and full-pipeline measurement error.

\begin{center}
\begin{minipage}{0.98\textwidth}
\centering
\captionof{table}{Joint-dominance sensitivity under the odds-scale multiplier. Panel A reports coverage at $\kappa=2.5$; Panel B reports the smallest grid value of $\kappa$ attaining at least 95\% coverage.}
\label{tab:app--dominance-diagnostics}
\begin{tabular}{lcccc}
\toprule
\multicolumn{5}{l}{\emph{Panel A: coverage at $\kappa=2.5$}}\\
Hidden odds multiplier $\delta$ & $p=10$ & $p=20$ & $p=40$ & $p=80$ \\
\midrule
1.0 & 0.998 & 0.999 & 1.000 & 0.999 \\
1.5 & 0.989 & 0.993 & 0.995 & 0.996 \\
2.0 & 0.963 & 0.975 & 0.984 & 0.984 \\
2.5 & 0.924 & 0.947 & 0.960 & 0.965 \\
3.0 & 0.879 & 0.907 & 0.925 & 0.939 \\
4.0 & 0.767 & 0.808 & 0.837 & 0.860 \\
\midrule
\multicolumn{5}{l}{\emph{Panel B: minimum $\kappa$ for at least 95\% coverage}}\\
Hidden odds multiplier $\delta$ & $p=10$ & $p=20$ & $p=40$ & $p=80$ \\
\midrule
1.0 & 1.25 & 1.25 & 1.0 & 1.0 \\
1.5 & 2.0 & 2.0 & 1.5 & 1.5 \\
2.0 & 2.5 & 2.5 & 2.0 & 2.0 \\
2.5 & 3.0 & 3.0 & 2.5 & 2.5 \\
3.0 & 4.0 & 4.0 & 3.0 & 3.0 \\
4.0 & 5.0 & 5.0 & 4.0 & 4.0 \\
\bottomrule
\end{tabular}
\end{minipage}
\end{center}

\begin{center}
\begin{minipage}{0.98\textwidth}
\centering
\captionof{table}{Alpha-side alignment and route-rule ablation. Panel A scales alpha-side benchmark odds while holding outcome-side signal fixed. Panel B evaluates the route rules at alpha multiplier 0.03.}
\label{tab:app--alpha-alignment-diagnostics}
\begin{tabular}{lcccc}
\toprule
\multicolumn{5}{l}{\emph{Panel A: alpha-side odds multiplier}}\\
Multiplier & \shortstack{Naive\\coverage} & \shortstack{Maximum-alpha\\gate} & \shortstack{Full\\gate} & \shortstack{Median alpha/\\outcome ratio} \\
\midrule
1.00 & 0.999 & 1.000 & 0.999 & 0.501 \\
0.30 & 0.980 & 1.000 & 0.965 & 0.152 \\
0.10 & 0.873 & 0.595 & 0.000 & 0.051 \\
0.03 & 0.564 & 0.006 & 0.000 & 0.015 \\
0.01 & 0.223 & 0.000 & 0.000 & 0.005 \\
\midrule
\multicolumn{5}{l}{\emph{Panel B: route-rule ablation at multiplier 0.03}}\\
Route rule & \shortstack{Promotion\\rate} & \shortstack{Coverage if\\promoted} & \multicolumn{2}{c}{ } \\
\midrule
No gate & 1.000 & 0.270 & \multicolumn{2}{c}{ } \\
Benchmark-count gate & 1.000 & 0.270 & \multicolumn{2}{c}{ } \\
Maximum-alpha gate & 0.936 & 0.273 & \multicolumn{2}{c}{ } \\
Ratio/effective-support gate & 0.000 & -- & \multicolumn{2}{c}{ } \\
Full route rule & 0.000 & -- & \multicolumn{2}{c}{ } \\
\bottomrule
\end{tabular}
\end{minipage}
\end{center}

\begin{center}
\begin{minipage}{0.98\textwidth}
\centering
\captionof{table}{Product effective support and SDID Riesz inflation in a single-treated-unit-sized panel. The final row is the empirical second-moment anchor.}
\label{tab:app--product-ess}
\begin{tabular*}{\textwidth}{@{\extracolsep{\fill}}lcccc@{}}
\toprule
\shortstack{Dirichlet\\concentration} & \shortstack{Median ESS\\share} & 10th pct. & 90th pct. & \shortstack{Median\\inflation} \\
\midrule
0.1 & 0.019 & 0.009 & 0.036 & 52.61 \\
0.3 & 0.076 & 0.043 & 0.122 & 13.11 \\
1.0 & 0.283 & 0.205 & 0.366 & 3.54 \\
3.0 & 0.587 & 0.506 & 0.658 & 1.70 \\
10.0 & 0.836 & 0.795 & 0.871 & 1.20 \\
20.0 & 0.912 & 0.888 & 0.932 & 1.10 \\
\midrule
\shortstack[l]{California tobacco-\\control anchor} & 0.177 & -- & -- & 5.64 \\
\bottomrule
\end{tabular*}
\end{minipage}
\end{center}

\begin{center}
\begin{minipage}{0.98\textwidth}
\centering
\captionof{table}{Refit finite-difference stability. Each row summarizes direction counts 40, 80, 160, and 320 at the listed step size.}
\label{tab:app--refit-grid}
\begin{tabular*}{\textwidth}{@{\extracolsep{\fill}}lcccc@{}}
\toprule
$h_{\mathrm{fd}}$ & \shortstack{Correlation\\range} & \shortstack{Median relative-\\difference range} & \shortstack{$M_{\mathrm{refit}}$\\range} & \shortstack{RV\\range} \\
\midrule
$10^{-4}$ & 0.941--0.942 & 0.392--0.399 & 338.0--342.8 & 0.0455--0.0461 \\
$3\times10^{-4}$ & 0.942--0.942 & 0.385--0.399 & 340.6--342.7 & 0.0455--0.0458 \\
$10^{-3}$ & 0.942--0.943 & 0.393--0.397 & 339.9--341.8 & 0.0456--0.0459 \\
$3\times10^{-3}$ & 0.941--0.942 & 0.392--0.403 & 338.7--342.1 & 0.0456--0.0461 \\
$10^{-2}$ & 0.939--0.941 & 0.399--0.404 & 338.4--343.5 & 0.0454--0.0461 \\
\bottomrule
\end{tabular*}
\end{minipage}
\end{center}

\begin{center}
\begin{minipage}{0.98\textwidth}
\centering
\captionof{table}{Auxiliary sensitivity in the calibrated full-update experiment.}
\label{tab:app--auxiliary-sensitivity}
\begin{tabular}{lcc}
\toprule
\multicolumn{3}{l}{\emph{Panel A: working family}}\\
Specification & Median predictive RV & $\Pr(RV<0.10)$ \\
\midrule
Beta marginal & 0.128 & 0.188 \\
Logit-normal & 0.121 & 0.226 \\
Gaussian copula & 0.134 & 0.201 \\
\midrule
\multicolumn{3}{l}{\emph{Panel B: link rule}}\\
Mean link & 0.071 & 0.742 \\
95th-percentile link & 0.121 & 0.199 \\
Maximum link & 0.184 & 0.046 \\
\midrule
\multicolumn{3}{l}{\emph{Panel C: prior concentration}}\\
10 & 0.119 & 0.226 \\
20 & 0.121 & 0.199 \\
50 & 0.126 & 0.191 \\
100 & 0.128 & 0.188 \\
\bottomrule
\end{tabular}
\end{minipage}
\end{center}

\begin{center}
\begin{minipage}{0.98\textwidth}
\centering
\captionof{table}{Single-treated-unit small-benchmark route stability. Panels report benchmark-count, alpha-support, and leave-one-benchmark-out route-stability checks.}
\label{tab:app--tobacco-demotion-stability}
\emph{Panel A: $p=7$ demotion stress}\par\smallskip
\begin{tabular}{lccccc}
\toprule
$p_{\mathrm{eff}}$ & \shortstack{Median\\$R_y^2$} & \shortstack{Median\\$R_\alpha^2$} & \shortstack{Count\\gate} & \shortstack{Alpha\\gate} & \shortstack{Naive\\coverage} \\
\midrule
7 & 0.455 & $1.0\times10^{-6}$ & 0.000 & 0.000 & 0.285 (0.004) \\
\bottomrule
\end{tabular}
\par\medskip
\emph{Panel B: benchmark-count threshold}\par\smallskip
\begin{tabular}{lcccc}
\toprule
Threshold & 5 & 7 & 10 & 20 \\
\midrule
Count-gate pass & 1 & 1 & 0 & 0 \\
Alpha-gate pass & 0 & 0 & 0 & 0 \\
\bottomrule
\end{tabular}
\par\medskip
\emph{Panel C: alpha-support threshold}\par\smallskip
\begin{tabular}{lcccc}
\toprule
Threshold & 0.001 & 0.003 & 0.005 & 0.010 \\
\midrule
Alpha-gate pass & 0 & 0 & 0 & 0 \\
Promotion rate & 0.000 & 0.000 & 0.000 & 0.000 \\
\bottomrule
\end{tabular}
\par\medskip
\emph{Panel D: leave-one-benchmark-out}\par\smallskip
\begin{tabular}{lcc}
\toprule
Dropped benchmark & Remaining count & Max $R_\alpha^2$ \\
\midrule
None & 7 & $1.0\times10^{-6}$ \\
GDP & 6 & $1.0\times10^{-6}$ \\
Pretreatment smoking mean & 6 & $1.0\times10^{-6}$ \\
Retail price & 6 & $1.0\times10^{-6}$ \\
Age 15--24 & 6 & $1.0\times10^{-6}$ \\
\bottomrule
\end{tabular}
\end{minipage}
\end{center}

\begin{center}
\begin{minipage}{0.98\textwidth}
\centering
\captionof{table}{Benchmark dependence and full-pipeline measurement error. Panel A compares nominal-$p$ and effective-count coverage under equicorrelation. Panel B reports a non-equicorrelation stress check for the 95th-percentile link. Panel C regenerates noisy covariates and recomputes partial-$R^2$ benchmarks.}
\label{tab:app--dependence-measurement}
\emph{Panel A: equicorrelated benchmark pairs, nominal $p=40$}\par\smallskip
\begin{tabular}{lccccc}
\toprule
Correlation & $p_{\mathrm{eff}}$ & \shortstack{Naive\\coverage} & \shortstack{Corrected\\coverage} & \shortstack{Naive\\width} & \shortstack{Corrected\\width} \\
\midrule
0.00 & 40.00 & 0.974 & 0.974 & 0.0418 & 0.0418 \\
0.25 & 3.72 & 0.927 & 0.970 & 0.0342 & 0.0474 \\
0.50 & 1.95 & 0.870 & 0.970 & 0.0271 & 0.0541 \\
0.75 & 1.32 & 0.785 & 0.973 & 0.0187 & 0.0610 \\
\bottomrule
\end{tabular}
\par\medskip
\emph{Panel B: non-equicorrelated score dependence, 95th-percentile link, 5,000 replications}\par\smallskip
\begin{tabular}{lcccccc}
\toprule
Design & \shortstack{Avg.\\corr.} & \shortstack{Mean\\$p_{\mathrm{eff}}$} & \shortstack{Spectral\\$p_{\mathrm{eff}}$} & \shortstack{Naive\\cov.} & \shortstack{Mean-eff.\\cov.} & \shortstack{Spectral\\cov.} \\
\midrule
Equicorr 0.50 & 0.500 & 1.95 & 3.72 & 0.496 & 1.000 & 0.986 \\
Block $8\times5$, 0.75 & 0.077 & 10.00 & 12.31 & 0.783 & 0.992 & 0.980 \\
Block $4\times10$, 0.75 & 0.173 & 5.16 & 6.60 & 0.636 & 0.996 & 0.988 \\
AR(1) 0.80 & 0.179 & 5.00 & 9.28 & 0.682 & 0.999 & 0.975 \\
AR(1) 0.90 & 0.348 & 2.75 & 4.76 & 0.545 & 0.999 & 0.984 \\
\bottomrule
\end{tabular}
\par\medskip
\emph{Panel C: full-pipeline measurement error, 1,500 replications}\par\smallskip
\begin{tabular}{lccccc}
\toprule
Reliability & \shortstack{Observed/latent\\$R_\alpha^2$} & \shortstack{Gate\\rate} & \shortstack{Conditional\\coverage} & \shortstack{Unconditional\\coverage} & \shortstack{Median\\reported $R^2$} \\
\midrule
1.00 & 1.000 & 0.963 & 0.983 & 0.976 & 0.0230 \\
0.90 & 0.976 & 0.955 & 0.983 & 0.977 & 0.0224 \\
0.75 & 0.938 & 0.935 & 0.974 & 0.965 & 0.0214 \\
0.50 & 0.812 & 0.819 & 0.945 & 0.925 & 0.0187 \\
0.25 & 0.613 & 0.537 & 0.875 & 0.787 & 0.0138 \\
\bottomrule
\end{tabular}
\end{minipage}
\end{center}

The small-$p$ design is Route A-primary for every count threshold, alpha threshold, and leave-one-benchmark-out specification in Table \ref{tab:app--tobacco-demotion-stability}. Table \ref{tab:app--dependence-measurement} shows two distinct limitations. First, redundant benchmark pairs overstate auxiliary information unless $p_{\mathrm{eff}}$ is used; the block and AR designs show that this is not only an equicorrelation issue. Panel B also shows the intended calibration tradeoff: the mean-score adjustment is conservative under strong non-equicorrelated dependence, whereas the spectral participation ratio is closer to nominal; in applications with suspected non-equicorrelation, we therefore report the spectral value alongside the mean-score value as a less conservative calibration check. Second, measurement error attenuates the observed alpha side and eventually degrades both promotion and conditional coverage.

\subsection{California Tobacco-Control Robustness Checks}
\label{app:tobacco}

\noindent\emph{Implementation and refit diagnostics.} The single-treated-unit study uses seven state-level covariates from the California tobacco-control dataset analyzed by \citet{AbadieDiamondHainmueller2010}. Table \ref{tab:app--tobacco-demotion-stability} reports three route-stability checks: minimum benchmark-count threshold sensitivity, alpha-support-threshold sensitivity, and leave-one-benchmark-out demotion stability. Because $p=7$ is below the minimum benchmark count for a calibrated auxiliary benchmark model, and because the observed state-level covariates generate nearly degenerate alpha-side benchmarks, the main text classifies the application as Route A-primary. The leave-one-control-out placebo loop removes California, treats each donor state once as pseudo-treated, and uses the remaining donor states as controls. This yields placebo SE 9.49 and add-one left-tail and absolute-rank $p$-values of $2/39=0.051$. Table \ref{tab:app-finite-diff-combined} reports the finite-difference derivative and decision-scale summaries.

\begin{center}
\begin{minipage}{0.98\textwidth}
\centering
\captionof{table}{California tobacco-control finite-difference refit diagnostic. The upper panel compares fixed-weight and refitted derivatives; the lower panel translates the same diagnostic into the Route A decision scale.}
\label{tab:app-finite-diff-combined}
\begin{tabular}{lll}
\toprule
Panel & Quantity & Value \\
\midrule
Derivative & Random perturbation directions & 80 \\
Derivative & Central-difference step $h_{\mathrm{fd}}$ & 0.001 \\
Derivative & Correlation, fixed-weight vs. refit derivative & 0.941 \\
Derivative & Median absolute difference & 0.167 \\
Derivative & 90th percentile absolute difference & 0.498 \\
Derivative & Median symmetric relative difference & 0.250 \\
Derivative & 90th percentile symmetric relative difference & 0.976 \\
\midrule
Decision scale & Fixed-weight exact $\E[\alpha^2]$ and $M$ & 567.83; 281.98 \\
Decision scale & Fixed-weight exact nullification RV & 0.054 \\
Decision scale & Fixed-weight random projection RV & 0.054 \\
Decision scale & Refit-weight projection $\E[\alpha^2]$ and $M$ & 823.98; 339.68 \\
Decision scale & Refit-weight projection nullification RV & 0.045 \\
Decision scale & RV to preserve 5\% placebo significance & $<0.001$ on all scales \\
\bottomrule
\end{tabular}
\end{minipage}
\end{center}

\medskip
\noindent\emph{Placebo-design comparison.} Table \ref{tab:app--placebo-route} reports the route-status results across the 38 donor pseudo-treatment assignments. The donor-only rank for the California tobacco-control placebo exercise is $1/38=0.026$ for both left-tail and absolute-rank calculations. The main text instead uses the full single-treated assignment rank, which counts California as the realized treated assignment and reports the add-one value $2/39=0.051$. The two calculations use the same placebo standard error, 9.49, and differ only in whether the realized treated assignment is counted as part of the rank experiment.

\begin{center}
\begin{minipage}{0.98\textwidth}
\centering
\captionof{table}{Donor-placebo route-stability diagnostic for the California tobacco-control panel. Each donor state is pseudo-treated once while California is excluded. The diagnostic checks whether the route rule promotes calibrated Route B in placebo assignments.}
\label{tab:app--placebo-route}
\begin{tabular}{lc}
\toprule
Quantity & Value \\
\midrule
Donor pseudo-treatment assignments & 38 \\
Assignments classified as calibrated Route B & 0 \\
Calibrated-route rate & 0.000 \\
Median nullification RV & 0.021 \\
90th percentile nullification RV & 0.061 \\
Maximum nullification RV & 0.086 \\
Median maximum alpha-side benchmark & $1.0\times 10^{-6}$ \\
\bottomrule
\end{tabular}
\end{minipage}
\end{center}

\medskip
\noindent\emph{Spillover-prone donor-block exclusion.} Because state-level tobacco policies may have cross-border or western-region spillovers, Table \ref{tab:app-spillover-block} re-estimates the SDID Route A scale after excluding Nevada and increasingly broad western donor blocks. The nullification RV stays in the same low-single-digit range, while the 5\% placebo-significance RV remains below 0.001 under the full-sample corrected placebo scale.

\begin{center}
\begin{minipage}{0.98\textwidth}
\centering
\captionof{table}{California tobacco-control donor-block exclusion diagnostic. The western block contains Nevada, Utah, New Mexico, Colorado, and Idaho.}
\label{tab:app-spillover-block}
\begin{tabular}{lccc}
\toprule
Specification & Donors & $\widehat\tau_{\mathrm{SDID}}$ / $M$ & Nullification RV \\
\midrule
Baseline donor pool & 38 & $-15.60$ / 281.98 & 0.054 \\
Exclude Nevada & 37 & $-17.06$ / 284.67 & 0.058 \\
Exclude NV/UT/NM & 35 & $-16.80$ / 280.24 & 0.058 \\
Exclude western block & 33 & $-16.21$ / 280.13 & 0.056 \\
\bottomrule
\end{tabular}
\end{minipage}
\end{center}

\medskip
\noindent\emph{Cinelli--Hazlett OLS special case.} The Riesz-general RV formula reduces to the linear-regression expression of \citet{CinelliHazlett2020} when $M=SE\sqrt{df}$. Table \ref{tab:app-cinelli-special} re-estimates the California tobacco-control panel with a two-way fixed-effect regression and computes the nullification RV using both forms.

\begin{center}
\begin{minipage}{0.98\textwidth}
\centering
\captionof{table}{California tobacco-control two-way fixed-effects OLS special-case check.}
\label{tab:app-cinelli-special}
\begin{tabular}{lccccc}
\toprule
Method & $\widehat\tau_{\mathrm{TWFE}}$ & SE & $t$ & Nullification RV & 5\% RV \\
\midrule
Cinelli-Hazlett OLS TWFE & $-27.35$ & 4.41 & $-6.20$ & 0.168 & 0.118 \\
Riesz-general OLS TWFE & $-27.35$ & 4.41 & $-6.20$ & 0.168 & 0.118 \\
\bottomrule
\end{tabular}
\end{minipage}
\end{center}

\FloatBarrier
\subsection{Staggered-Adoption Robustness Checks}

Table \ref{tab:app--minwage-leave-cohort} reports a leave-cohort-out design check. The aggregate estimate varies with the omitted cohort, but the Route A nullification RV remains high in every deletion.

\begin{center}
\begin{minipage}{0.98\textwidth}
\centering
\captionof{table}{Minimum-wage leave-cohort-out design check.}
\label{tab:app--minwage-leave-cohort}
\begin{tabular}{lcccc}
\toprule
Specification & Weighted ATT & SE & Nullification RV & 5\% RV \\
\midrule
All cohorts & $-0.040$ & 0.012 & 0.993 & 0.958 \\
Drop 2004 cohort & $-0.033$ & 0.014 & 0.984 & 0.921 \\
Drop 2006 cohort & $-0.052$ & 0.017 & 0.996 & 0.970 \\
Drop 2007 cohort & $-0.041$ & 0.013 & 0.991 & 0.949 \\
\bottomrule
\end{tabular}
\end{minipage}
\end{center}

Figure \ref{fig:app-staggered-attgt} plots the group-time $ATT(g,t)$ cells underlying the staggered-adoption study, using the \texttt{mpdta} panel considered by \citet{CallawaySantAnna2021}. The fixed-effect-imputation comparison reported in the main-text minimum-wage table fits two-way fixed effects on untreated cells and computes the same coefficient-scale RMS Riesz norm on the balanced county-year panel. It gives an imputation estimate of $-0.048$ and Riesz norm 0.00203, about 0.90 times the aggregate group-time norm. This diagnostic comparison shows that the group-time and fixed-effect-imputation summaries are close on both effect and Riesz scales.

\begin{center}
\begin{minipage}{0.98\textwidth}
\centering
\includegraphics[width=0.80\textwidth]{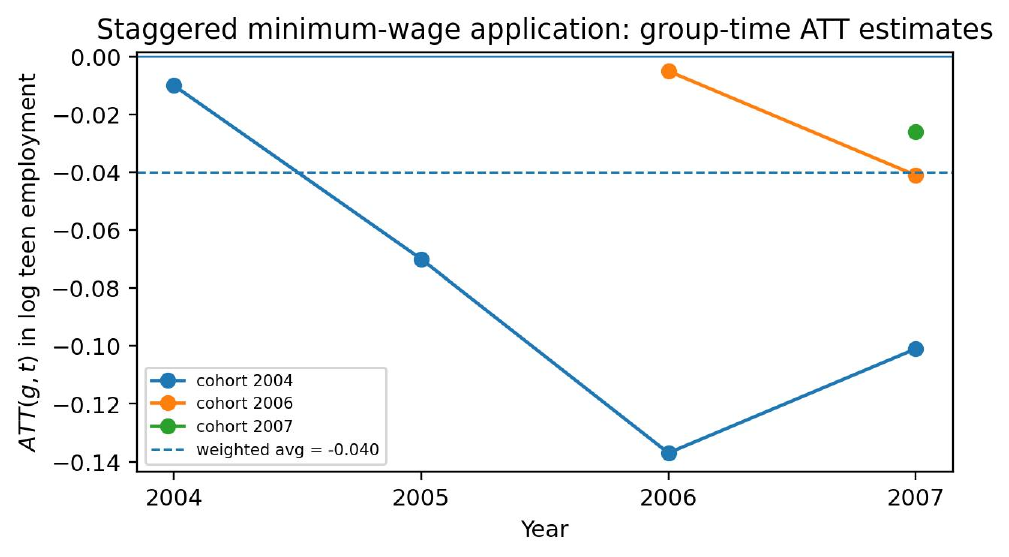}
\captionof{figure}{Group-time $ATT(g,t)$ cells in the staggered-adoption study. Points report cohort-year contrasts against never-treated counties; the horizontal line marks the cohort-size-weighted aggregate.}
\label{fig:app-staggered-attgt}
\end{minipage}
\end{center}

\clearpage

\FloatBarrier
\clearpage
\section*{Data Availability Statement}
The empirical study in Section~\ref{sec:application} uses the annual state-level cigarette-consumption panel analyzed by \citet{AbadieDiamondHainmueller2010}. The analysis covers California and 38 untreated donor states from 1970 through 2000. California is treated from 1989, the outcome is cigarette packs per capita, and the seven benchmark variables are log GDP per capita, pretreatment smoking mean, the population share aged 15--24, retail cigarette price, and cigarette consumption in 1980, 1975, and 1988. The empirical study in Section~\ref{sec:staggered-app} uses the county-level \texttt{mpdta} minimum-wage panel analyzed by \citet{CallawaySantAnna2021}. That panel contains 500 counties observed from 2003 through 2007, log teen employment as the outcome, three treatment cohorts, and 309 never-treated counties. No new data were collected for this study. Data provenance and access conditions are described by the original sources cited in Sections~\ref{sec:application} and~\ref{sec:staggered-app}.

\section*{Funding}
No financial support was received for this research.

\section*{Disclosure Statement}
No competing interests are declared.

\bibliographystyle{plainnat}
\bibliography{references}

\end{document}